\documentclass[conference]{IEEEtran}
\usepackage{url}
\usepackage{hyperref}
\hypersetup{
  colorlinks   = true, 
  linkcolor    = black!50!brown, 
  citecolor   = black!50!brown, 
}
\usepackage{balance}
\usepackage{verbatim}
\usepackage{bm}
\usepackage{color,graphicx}
\usepackage[table]{xcolor}
\usepackage{colortbl}
\usepackage{mdwtab}
\usepackage{mathtools,tikz}
\usepackage{pgfplots}
\pgfplotsset{soldot/.style={color=myblue,only marks,mark=*}} \pgfplotsset{holdot/.style={color=myblue,fill=white,only marks,mark=*}}
\usepackage{hhline}
\usepackage{multirow}
\usepackage{array}
\usepackage{pdfpages}
\usepackage{enumitem}  
\usepackage{caption}
\usepackage{subcaption}
\DeclareCaptionFormat{myformat}{\fontsize{8}{9}\selectfont#1#2#3}
\usepackage{cite}
\usepackage{amsmath,amssymb,amsfonts,amsthm,xspace,bm}
\usepackage{cleveref}
\usepackage[bb=dsserif]{mathalpha}
\usepackage{algorithm}
\usepackage{algorithmic}
\usepackage{graphicx}
\usepackage{textcomp}
\usepackage{todonotes}
\usepackage{etoolbox}
\usepackage{multicol}
\usetikzlibrary{arrows,shapes.geometric,calc}
\usepackage[rightcaption]{sidecap}
\sidecaptionvpos{figure}{m}
\allowdisplaybreaks
\definecolor{bluishgreen}{RGB}{0,158,115}
\definecolor{vermillion}{RGB}{213,94,0}
\definecolor{myblue}{RGB}{0,114,178}
\definecolor{myorange}{RGB}{230,159,0}
\definecolor{offwhite}{RGB}{245,245,220}

\theoremstyle{definition}
\newtheorem{example}{Example}
\newtheorem{remark}{Remark}
\newtheorem{theorem}{Theorem}
\newtheorem{definition}{Definition}

\newtheorem{claim}[theorem]{Claim}

\newtheorem*{definition*}{Definition}
\newtheorem*{lemma*}{Lemma}

\newtheorem{lemma}[theorem]{Lemma}

\newtheoremstyle{thmnum}{\topsep}{\topsep}{\itshape}{0pt}{\bfseries}{.}{ }{\thmname{#1}\thmnote{ \bfseries #3}}
\theoremstyle{thmnum}

\usepackage{pgffor}
\foreach \x in {a,...,z}{%
\expandafter\xdef\csname vec\x \endcsname{\noexpand\ensuremath{\noexpand\bm{\x}}}
}

\foreach \x in {A,...,Z}{%
\expandafter\xdef\csname vec\x \endcsname{\noexpand\ensuremath{\noexpand\bm{\x}}}
}

\foreach \x in {A,...,Z}{%
\expandafter\xdef\csname c\x \endcsname{\noexpand\ensuremath{\noexpand\mathcal{\x}}}
}

\foreach \x in {A,...,Z}{%
\expandafter\xdef\csname bb\x \endcsname{\noexpand\ensuremath{\noexpand\mathbb{\x}}}
}

\foreach \x in {A,...,Z}{%
\expandafter\xdef\csname s\x \endcsname{\noexpand\ensuremath{\noexpand\sf{\x}}}
}

\foreach \x in {J,M,X,Y,Z}{%
\expandafter\xdef\csname h\x \endcsname{\noexpand\ensuremath{\noexpand\widehat{\x}}}
}

\foreach \x in {m,x,y,z}{%
\expandafter\xdef\csname h\x \endcsname{\noexpand\ensuremath{\noexpand\hat{\x}}}
}

\foreach \x in {m,x,y,z}{%
\expandafter\xdef\csname t\x \endcsname{\noexpand\ensuremath{\noexpand\tilde{\x}}}
}
\xdef\tzero{\noexpand\ensuremath{\noexpand\tilde{0}}}
\xdef\tone{\noexpand\ensuremath{\noexpand\tilde{1}}}

\foreach \x in {M,P,X,Y,Z,x,y,z}{%
\expandafter\xdef\csname b\x \endcsname{\noexpand\ensuremath{\noexpand\bar{\x}}}
}

\newcommand{\p}{\sf P}
\newcommand{\party}[1]{{\p}_{#1}}
\newcommand{\pa}{\party{1}}
\newcommand{\pb}{\party{2}}
\newcommand{\pc}{\party{3}}

\newcommand{\secretalphabet}[1]{\cX_{#1}}
\newcommand{\cXa}{\secretalphabet{1}}
\newcommand{\cXb}{\secretalphabet{2}}
\newcommand{\cXc}{\secretalphabet{3}}
\newcommand{\cXi}{\secretalphabet{i}}

\newcommand{\X}{X}  
\newcommand{\secretrv}[1]{\X_{#1}}
\newcommand{\Xa}{\secretrv{1}}
\newcommand{\Xb}{\secretrv{2}}
\newcommand{\Xc}{\secretrv{3}}
\newcommand{\Secrets}{\bm{\X}}

\newcommand{\x}{x}
\newcommand{\secret}[1]{\x_{#1}}
\newcommand{\xa}{\secret{1}}
\newcommand{\xb}{\secret{2}}
\newcommand{\xc}{\secret{3}}
\newcommand{\secrets}{\bm{\x}}

\newcommand{\salp}{\cW}
\newcommand{\sharealphabet}[2]{\salp_{#1#2}}
\newcommand{\salpab}{\sharealphabet{1}{2}}
\newcommand{\salpbc}{\sharealphabet{2}{3}}
\newcommand{\salpca}{\sharealphabet{3}{1}}
\newcommand{\salpij}{\sharealphabet{i}{j}}
\newcommand{\twosharealphabet}[1]{\salp_{#1}}
\newcommand{\salpa}{\twosharealphabet{1}}
\newcommand{\salpb}{\twosharealphabet{2}}
\newcommand{\salpc}{\twosharealphabet{3}}
\newcommand{\salpi}{\twosharealphabet{i}}

\newcommand{\Share}{W}
\newcommand{\sharerv}[2]{\Share_{#1#2}}
\newcommand{\Sab}{\sharerv{1}{2}}
\newcommand{\Sbc}{\sharerv{2}{3}}
\newcommand{\Sca}{\sharerv{3}{1}}
\newcommand{\Sij}{\sharerv{i}{j}}
\newcommand{\twosharerv}[1]{W_{#1}}
\newcommand{\Sa}{\twosharerv{1}}
\newcommand{\Sb}{\twosharerv{2}}
\newcommand{\Sc}{\twosharerv{3}}
\newcommand{\Si}{\twosharerv{i}}
\newcommand{\Shares}{\bm{\Share}}

\newcommand{\share}{w}

\newcommand{\twoshareinst}[1]{w_{#1}}
\newcommand{\ssa}{\twoshareinst{1}}

\newcommand{\ssi}{\twoshareinst{i}}
\newcommand{\shares}{\bm{\share}}

\newcommand{\scheme}{\Pi}
\newcommand{\encoder}{\psi}
\newcommand{\decoder}{\phi}

\newcommand{\decoderplayer}[1]{\decoder_{#1}}
\newcommand{\decodera}{\decoderplayer{1}}
\newcommand{\decoderb}{\decoderplayer{2}}
\newcommand{\decoderc}{\decoderplayer{3}}

\newcommand{\rc}{\rho}
\newcommand{\pins}{\cP_{\cS}}

\newcommand{\A}{a}
\newcommand{\symbola}[1]{\A_{#1}}
\newcommand{\aone}{\symbola{1}}
\newcommand{\atwo}{\symbola{2}}
\newcommand{\athree}{\symbola{3}}
\newcommand{\afour}{\symbola{4}}

\newcommand{\B}{b}
\newcommand{\symbolb}[1]{\B_{#1}}
\newcommand{\bone}{\symbolb{1}}
\newcommand{\btwo}{\symbolb{2}}
\newcommand{\bthree}{\symbolb{3}}
\newcommand{\bfour}{\symbolb{4}}
\newcommand{\bfive}{\symbolb{5}}
\newcommand{\bsix}{\symbolb{6}}

\newcommand{\C}{c}
\newcommand{\symbolc}[1]{\C_{#1}}
\newcommand{\cone}{\symbolc{1}}
\newcommand{\ctwo}{\symbolc{2}}
\newcommand{\cthree}{\symbolc{3}}
\newcommand{\cfour}{\symbolc{4}}
\newcommand{\cfive}{\symbolc{5}}
\newcommand{\csix}{\symbolc{6}}

\newcommand{\Prob}{\mathbb{P}}

\def\BibTeX{{\rm B\kern-.05em{\sc i\kern-.025em b}\kern-.08em
    T\kern-.1667em\lower.7ex\hbox{E}\kern-.125emX}}

\usepackage[skip=0.3em]{caption}
\DeclareCaptionFormat{myformat}{\fontsize{8}{9}\selectfont#1#2#3}
\IEEEoverridecommandlockouts
\begin{document}
\title{Randomness Requirements for\\ Three-Secret Sharing}


\makeatletter
\newcommand{\linebreakand}{%
  \end{@IEEEauthorhalign}
  \hfill\mbox{}\par
  \mbox{}\hfill\begin{@IEEEauthorhalign}
}
\makeatother

\author{\IEEEauthorblockN{Hari Krishnan P. Anilkumar}
\IEEEauthorblockA{TIFR, India}
\and
\IEEEauthorblockN{Aayush Rajesh}
\IEEEauthorblockA{IIT Bombay, India}
\and
\IEEEauthorblockN{Varun Narayanan}
\IEEEauthorblockA{Technion, Israel}
\and
\IEEEauthorblockN{Manoj M. Prabhakaran}
\IEEEauthorblockA{IIT Bombay, India}
\and
\IEEEauthorblockN{Vinod M. Prabhakaran}
\IEEEauthorblockA{TIFR, India}

\thanks{HA and VP were supported by DAE under project no. RTI4001. VN was supported by ERC Project NTSC (742754) and ISF Grants 1709/14 and 2774/20. MP was supported by IITB Trust Lab. VP was additionally supported by SERB through project MTR/2020/000308.}
}

\maketitle

\begin{abstract}
We study a secret sharing problem with three secrets where the secrets are allowed to be related to each other, i.e., only certain combinations of the three secrets are permitted. The dealer produces three shares such that every pair of shares reveals a unique secret and reveals nothing about the other two secrets, other than what can be inferred from the revealed secret. For the case of binary secrets, we exactly determine the minimum amount of randomness required by the dealer, for each possible set of permitted combinations. Our characterization is based on new lower and upper bounds.
\end{abstract}

\section{Introduction}\label{sec:introduction}
In this paper we study Multi-Secret Sharing (MSS), a variation on the secret sharing problem in which multiple secrets are shared simultaneously. The secrets are a priori known to be {\em related} in the sense that they are guaranteed to come from a known domain of secrets. Specifically, we study the following 3-secret, 3-shares problem, denoted as 3SS: There are three secrets with alphabets $\cX, \cY, \cZ$ respectively.
The secrets $(\xa,\xb,\xc)$ are known to be from a subset $\cS\subseteq\cX\times\cY\times\cZ$. The dealer produces three different shares $\Sab,\Sbc,\Sca$ such that every pair of shares reveals a unique secret and nothing more. Specifically, shares $\Sab,\Sca$ must reveal the secret $\xa$, but no additional information about $\xb,\xc$ other than what can be inferred from learning $\xa$ and the fact that the secrets belong to $\cS$. Similarly, shares $\Sab,\Sbc$ must reveal $\xb$ and nothing more; and $\Sbc,\Sca$ must recover $\xc$ and nothing more.  Clearly, on learning all three shares, all the three secrets will be recoverable. There are no restrictions on what a single share can reveal except what is implied by the restrictions on pairs of shares.

Each subset $\cS$ corresponds to a different secret-sharing problem.  There
are two fundamental measures of efficiency of a secret-sharing scheme: the
sizes of the individual shares, and the overall amount of randomness
required by the dealer to prepare the shares.  The latter, called the
randomness complexity, can also be interpreted as the amount of information,
in addition to the secrets, that the dealer needs to store if it were to
reproduce the same shares at a later point.%
\footnote{This would be relevant in a setting where the shares are left on
the ``cloud'' and downloaded asynchronously by users at various points in
time.}
In contrast to much of the work on secret-sharing, our focus is on the
randomness complexity. Our main contribution is to \emph{completely resolve}
the question of randomness complexity of 3SS for
the most basic setting with binary secrets: We do this for 
\emph{every possible
domain} with binary secrets -- i.e., for each
$\cS\subseteq\{0,1\}\times\{0,1\}\times\{0,1\}$. This involves coming up
with new schemes, as well as new lower bounds.


Secret sharing with multiple secrets has been studied in the literature in various contexts\cite{CACM:McESar81,IT:KaGrHe83,C:BlaMea84,STACS:BluDeSVacc93,C:JacMarOKe93,C:BDDGV94,AC:JacMarOKe94,JacMarOKe96,Mas06,IT:DaPrPr16}. Some of the earlier works considered the setting where all secrets have the same access structure\cite{CACM:McESar81,IT:KaGrHe83,C:BlaMea84} (and hence the multiple secrets can be considered a single larger secret) and identified that ``ramp'' access structures -- which allow sets which are neither ``qualified'' nor ``forbidden'' -- as a means to keep the share sizes from growing compared to when sharing a single secret. We shall be concerned with a more general notion, where the different secrets can have different access structures, as was studied in later works. There have been several definitions of security for multi-secret sharing schemes, all of which considered secrets as random variables. But as shown by Masucci \cite{Mas06}, the security definitions used in the prior work -- the restrictive ``single-secret security'' (which could hold only if the secrets were independent of each other) and the more generally applicable ``multiple-secrets security'' -- implicitly depended only on the support of the joint distribution of the secrets. In this work we explicitly reformulate the latter definition solely in terms of the domain of the secrets, without any reference to a distribution of the secrets.



We point out that 3SS considers a specific access structure. One of our motivations for studying this model is that it arises naturally in the context of secure multi-party computation (MPC). In particular, in~\cite{HaNaPrPr23}, using a reduction of the problem at hand for a particular domain of secrets ($\cS=\{0,1\}^3 \setminus \{111\}$) to the secure multi-party computation (MPC) of the boolean function {\tt AND} by three parties with 1-privacy, we obtain a randomness complexity lower bound of 3~bits for secure MPC of {\tt AND} which improves the bound of $1$ bit due to Kushilevitz et al.~\cite{DM:KuOsEmRoThVe21} (who also showed an upper bound of 7 bits).  


The work which is most closely related to the present one is~\cite{IT:DaPrPr16}, where information-theoretic bounds on share sizes and randomness complexity for a variant of 3SS called {\em correlated multi-secret sharing} in which the secrets are distributed according to a given distribution were developed (see \cite[Appendix~B]{IT:DaPrPr16}). Those bounds were obtained for specific forms of the joint distribution of the secrets -- the choice of these distributions were motivated by a three-party secure MPC problem where two of the parties have inputs (and no outputs) and the third party securely computes a function of these inputs. The techniques we use to obtain information-theoretic lower bounds on randomness complexity in section~\ref{sec:lowerbounds} are similar to those used there.


For each $\cS\subseteq\{0,1\}\times\{0,1\}\times\{0,1\}$, we provide 3SS schemes and show their optimality in terms of randomness complexity (Section~\ref{sec:binary}). In all but two cases, our proof of optimality is based on a general information-theoretic lower bound on randomness complexity (Section~\ref{sec:lowerbounds}). For the remaining two cases where this lower bound is loose, we give a combinatorial argument. 


\section{Problem Setup}\label{sec:setup}
\begin{definition} \label{def:distribution-scheme}
Let $\cS\subseteq\cXa\times\cXb\times\cXc$ be a set of secrets, where $\cXa,\cXb,\cXc$ are finite. 
A {\em distribution scheme} $\scheme = (P_R,\encoder)$ with domain of secrets $\cS$ is a pair such that $P_R$ is a probability distribution over a finite set $\cR$, 
and $\encoder$ is a map $\encoder\colon\cS \times \cR \to \salpab\times\salpbc\times\salpca$, where $\salpij$ is the domain of the share $\Sij$, $(i,j)\in\{(1,2),(2,3),(3,1)\}$.
Given a secret $(\xa,\xb,\xc)\in\cS$, the dealer distributes it according to $\scheme$ by sampling $R$ according to $P_R$ and computing $(\Sab,\Sbc,\Sca)=\encoder((\xa,\xb,\xc),R)$. i.e., denoting $\Shares=(\Sab,\Sbc,\Sca)$, $\Secrets=(\Xa,\Xb,\Xc)$, the dealer generates the shares according to the conditional distribution
\begin{align} 
P_{\Shares|\Secrets}(\shares|\secrets) &= P_R(\{r\in\cR:\encoder(\secrets,r)=\shares\}), \label{eq:dealer-distr}
\end{align}
for $\secrets\in\cS$ and $\shares\in\salpab\times\salpbc\times\salpca$.

Party $\pa$ has access to $\Sa:=(\Sab,\Sca)$, party $\pb$ to $\Sb:=(\Sbc,\Sab)$, and party $\pc$ to $\Sc:=(\Sca,\Sbc)$.
Let $\salpa:=\salpab\times\salpca$ denote the alphabet of the shares party~$\pa$ has access to. Similarly, let
$\salpb:=\salpbc\times\salpab$ and $\salpc:=\salpca\times\salpbc$.
\end{definition} 

\begin{definition}\label{def:rss}
A distribution scheme $\scheme=(P_R,\encoder)$ is said to be a {\em  three-secret sharing (3SS) scheme} for its domain of secrets $\cS$ if the following hold:\\
{\bf Correctness.} For each party $\p_i$, $i\in[3]$, where $[3]:=\{1,2,3\}$, there is a reconstruction function $\decoder_i\colon \salpi \to \cXi$ such that, for $(\xa,\xb,\xc)\in\cS$, 
\begin{align} \Prob(\decoder_i(\Si)=\x_i)=1,\label{eq:correct} \end{align}
where the probability is over the random choice of $R\sim P_R$.\\
{\bf Perfect Privacy.} For each party $\p_i$, $i\in[3]$ and any two secret vectors $\secrets=(\xa,\xb,\xc)\in\cS$ and $\secrets'=(\xa',\xb',\xc')\in\cS$ such that $\x_i=\x_i'$, the shares $\Si$ that party $\p_i$ has access to must be identically distributed, i.e., (recall that  $\Sa:=(\Sab,\Sca)$ and so on)
\begin{align} P_{\Si|\Secrets}(\ssi|\secrets) =  P_{\Si|\Secrets}(\ssi|\secrets'), \quad \ssi\in\salpi.\label{eq:privacy-Pi}\end{align}
We refer to this as {\em privacy against party} $\p_i$.
\end{definition}

\begin{definition}\label{def:complexity}
The randomness complexity of a 3SS scheme $\scheme=(P_R,\encoder)$ is $\log|\cR|$, the number of bits of randomness it uses. For a set $\cS\subseteq\cXa\times\cXb\times\cXc$, the {\em randomness complexity} $\rho({\cS})$ is defined as the smallest randomness complexity of all 3SS schemes with domain of secrets $\cS$. We write $\rho$ when the domain of secrets is clear from the context.
\end{definition}

The following lemma gives an alternative definition for 3SS. The proof, which is elementary, is omitted in the interest of space.
\begin{lemma}\label{lem:ITdefinition}
A distribution scheme $\scheme=(P_R,\encoder)$ is a 3SS scheme for its domain of secrets $\cS$ if and only if for all joint distributions $P_{\Secrets}$ whose support is a subset of $\cS$ and $P_{\Shares|\Secrets}$ of \eqref{eq:dealer-distr}, the following hold (recall that $\Sa:=(\Sab,\Sca)$ and so on):
\begin{align}
&\text{(Correctness)} & H(\X_i|\Si)&=0,  &&i\in[3], \label{eq:ITcorrect}\\
&\text{(Privacy)} & I(\Si;\X_{[3]\setminus\{i\}}|\X_i)&=0, &&i\in[3]. \label{eq:ITprivacy}
\end{align}
\end{lemma}


\section{Information-Theoretic Lower Bounds}\label{sec:lowerbounds}

\begin{definition}[\cite{IT:WoWu08} (also see \cite{IT:PrPr14})]\label{def:RI}
    The {\em residual information} of jointly distributed random variables $X,Y$ is defined as
    \begin{align}
        RI(X;Y)&:= \min_{\substack{p_{Q|X,Y}: H(Q|X)=H(Q|Y)=0}} I(X;Y|Q). \label{eq:RIdef}
    \end{align}
\end{definition}
This is in fact the difference of $I(X;Y)$ and the G\'acs-K\"orner's common information~\cite{GaKo73} of $X,Y$. In \eqref{eq:RIdef},  an optimal choice for $Q$ is the maximal common random variable of G\'acs-K\'orner (i.e., the largest entropy random variable which can be written as a function of $X$ and a function of $Y$)~\cite{IT:WoWu08,IT:PrPr14}. 

Let $\pins=\{P_{\Secrets}:P_{\Secrets}(\secrets)=0 \text{ if } \secrets\notin\cS\}$ be the set of distributions $P_{\Secrets}$ whose support is a subset of $\cS$. For $i\in[3]$ and a distribution $P$ defined on $\cXi$, we define
\begin{align*}
\pins^i(P):=\left\{P_{\Secrets} \in \pins : P_{\X_i}=P\right\},
\end{align*}
i.e., the set of all distributions $P_{\Secrets}$ whose support is a subset of $\cS$ and the marginal $P_{\X_i}$ is $P$. 

The following is our information-theoretic lower bound on randomness complexity. It holds for all 3SS schemes in Definition~\ref{def:rss} (and not just when the secrets are binary).
\begin{theorem}\label{thm:RClowerbound}
Let $\scheme=(P_R,\encoder)$ be a 3SS scheme with domain of secrets $\cS$. 
\begin{align}
&\rc \geq\notag\\
&\sup_{P_{\Secrets}\in\pins} \Bigg\{\left[\sup_{P_{\Secrets'}\in\pins^2(P_{\Xb})} H(\Xa',\Xb'|\Xc') + RI(\Xb';\Xc')\right]\notag\\
&\qquad\quad\;+\left[\sup_{P_{\Secrets''}\in\pins^1(P_{\Xa})} H(\Xa'',\Xc''|\Xb'') + RI(\Xa'';\Xb'')\right]\notag\\
&\qquad\quad\;-  H(\Xa)\Bigg\}
\label{eq:LB2}\\
&\rc \geq\notag\\
&\sup_{P_{\Secrets}\in\pins} \Bigg\{ \left[\sup_{P_{\Secrets'}\in\pins^1(P_{\Xa})} H(\Xa',\Xb'|\Xc') + RI(\Xa';\Xc')\right]\notag\\
&\qquad\quad\;+\left[\sup_{P_{\Secrets''}\in\pins^1(P_{\Xa})} H(\Xa'',\Xc''|\Xb'') + RI(\Xa'';\Xb'')\right]\notag\\
&\qquad\quad\;-  H(\Xa)\Bigg\}
\label{eq:LB1}
\end{align}
\end{theorem}
Note that $RI(X;Y)$ is not a continuous function of the joint distribution $P_{X,Y}$. Hence, the $\sup$ above may not be replaced by $\max$ in general. 
The bounds in \eqref{eq:LB1} and \eqref{eq:LB2} are clearly asymmetric; permuting the indices $(1,2,3)$ give additional bounds on $\rc$.

\begin{remark}\label{rem:tensorize}
Consider 3SS for a domain of secrets which is an $n$-fold product $\cS^n$. Clearly, $\rc(\cS^n)\leq n\rc(\cS)$ since $n$ independent invocations of a 3SS for $\cS$ is a 3SS for $\cS^n$. The fact that conditional entropy and residual information tensorize\footnote{By ``tensorize'' we mean that the quantities for independent instances is the sum of the quantities for each instance. e.g., if $(U_1,V_1)$ and $(U_2,V_2)$ are independent, $H(U_1,U_2|V_1,V_2)=H(U_1|V_1)+H(U_2|V_2)$ and $RI(U_1,U_2;V_1,V_2)=RI(U_1;V_1)+RI(U_2;V_2)$.}~\cite{IT:WoWu08} implies that $n$~times the lower bounds in Theorem~\ref{thm:RClowerbound} for $\rc(\cS)$ are lower bounds for $\rc(\cS^n)$. Thus, the lower bounds in Theorem~\ref{thm:RClowerbound} are in fact lower bounds on $\frac{1}{n}\rc(\cS^n)$ for any $n$ and in particular as $n\rightarrow\infty$. Whenever it is tight for the single-shot case, we may conclude that $\rc(\cS^n)=n\rc(\cS)$ for all $n$. See the discussion in Section~\ref{sec:discussion} on an interesting question we leave open in this connection.
\end{remark}
\begin{proof}
We show \eqref{eq:LB2} here and prove \eqref{eq:LB1} in Appendix~\ref{app:lowerbound} along similar lines. Let $\scheme=(P_R,\encoder)$ be a 3SS scheme with domain of secrets $\cS$. Fix a $P_{\Secrets}\in\pins$.
These joint distribution over the secrets and shares is
\begin{align}
P_{\Secrets,\Shares}(\secrets,\shares)&= P_{\Secrets}(\secrets) P_R(\{r\in\cR:\encoder(\secrets,r)=\shares\}).\label{eq:joint-dist}
\end{align}
\begin{claim} \label{cl:firstLB1}
\begin{align} \log|\cR| \geq H(\Sab|\Sbc) + H(\Sca|\Sab) - H(\Xa)\label{eq:firstLB2}\end{align}
\end{claim}
\begin{proof}
\begin{align*}\log|\cR| &\geq H(R) = H(R|\Secrets)\stackrel{(a)}{=} H(R,\Shares|\Secrets),\end{align*}
where (a) follows from $\Shares=\encoder(\Secrets,R)$. Further,  $H(R,\Shares|\Secrets)$ $\geq H(\Shares|\Secrets) \geq H(\Sab,\Sca|\Secrets) \stackrel{(b)}{=} H(\Sab,\Sca|\Xa)$, where (b) follows from the privacy condition \eqref{eq:ITprivacy} for $i=1$.
But, $H(\Sab,\Sca|\Xa) = H(\Sab,\Sca,\Xa) - H(\Xa) \stackrel{(c)}{=} H(\Sab,\Sca) - H(\Xa)=H(\Sab) + H(\Sca|\Sab) - H(\Xa)$,
where (c) follows from the correctness condition \eqref{eq:ITcorrect} for $i=1$ which states that $H(\Xa|\Sab,\Sca)=0$.
Now \eqref{eq:firstLB2} follows from $H(\Sab) \geq H(\Sab|\Sbc)$.
\end{proof}
To prepare for lower bounding the conditional entropy terms in 
\eqref{eq:firstLB2}, we prove the following claim in Appendix~\ref{app:lowerbound} using a ``secure data processing inequality'' for residual information.
\begin{claim} \label{cl:useRI}
\begin{align}
    I(\Sab;\Sbc|\Sca) \geq RI(\Xa;\Xc), \label{eq:useRI1}\\
    I(\Sca;\Sbc|\Sab) \geq RI(\Xa;\Xb), \label{eq:useRI2}\\
    I(\Sab;\Sca|\Sbc) \geq RI(\Xb;\Xc). \label{eq:useRI3}
\end{align}
\end{claim}
To lower bound $H(\Sca|\Sab)$ in \eqref{eq:firstLB2}, observe that
\begin{align}
H&(\Sca|\Sab)\notag\\
 &= H(\Sca|\Sbc,\Sab) + I(\Sca;\Sbc|\Sab)\notag\\
 &\stackrel{\text{(a)}}{=} H(\Xa,\Xc,\Sca|\Sbc,\Sab,\Xb) + I(\Sca;\Sbc|\Sab)\notag\\
 &\geq H(\Xa,\Xc|\Sbc,\Sab,\Xb) + I(\Sca;\Sbc|\Sab)\notag\\
 &\stackrel{\text{(b)}}{=} H(\Xa,\Xc|\Xb) + I(\Sca;\Sbc|\Sab),\notag\\
 &\stackrel{\text{(c)}}{=} H(\Xa,\Xc|\Xb) + RI(\Xa;\Xb),\label{eq:bound-cond-entropy2}
\end{align}
where (a) follows from \eqref{eq:ITcorrect}, (b) follows from \eqref{eq:ITprivacy}, and (c) from \eqref{eq:useRI2}.
Similarly,
\begin{align}
H(\Sab|\Sbc) &\geq H(\Xa,\Xb|\Xc) +  RI(\Xb;\Xc).\label{eq:bound-cond-entropy3}
\end{align}
Since the choice of $P_{\Secrets}\in\pins$ was arbitrary, applying \eqref{eq:bound-cond-entropy2} and \eqref{eq:bound-cond-entropy3} in \eqref{eq:firstLB2}, we may conclude that 
\begin{align*}
\rc &\geq \sup_{P_{\Secrets}\in\pins} \Big(H(\Xa,\Xb|\Xc) + RI(\Xb;\Xc)\\
  &\qquad\qquad\qquad +  H(\Xa,\Xc|\Xb) + RI(\Xa;\Xb) -  H(\Xa)\Big).
\end{align*}
%

Notice that the above lower bound on $\rc$ is potentially weaker than the desired \eqref{eq:LB2} in the sense that it lacks the inner supremizations. 
We now strengthen the lower bounds in \eqref{eq:bound-cond-entropy2}-\eqref{eq:bound-cond-entropy3} to prove the theorem. We illustrate this for $H(\Sca|\Sab)$. The key idea behind the strengthening is the following observation (proved in Appendix~\ref{app:lowerbound}) which follows from privacy against $\pa$ (i.e., \eqref{eq:privacy-Pi} with $i=1$). It says that the joint distribution of the shares $(\Sab,\Sca)$ seen by party $\pa$ remains unchanged even if the distribution $P_{\Secrets}$ is replaced by another distribution $P_{\Secrets'}$ which has the same marginal distribution for the secret meant for $\pa$, i.e., $P_{\Xa'}=P_{\Xa}$.
\begin{claim}\label{cl:switch-distr}
Let $\scheme=(P_R,\encoder)$ be a 3SS scheme with domain of secrets $\cS$ and let $P_{\Secrets}\in\pins$ which together induce the joint distribution \eqref{eq:joint-dist}.
Let $P_{\Secrets'}\in\pins^1(P_{\Xa})$, i.e., $P_{\Secrets'}$ is another distribution whose support is a subset of $\cS$ with $P_{\Xa'}=P_{\Xa}$. Let the shares generated by $\scheme$ under $P_{\Secrets'}$ be denoted by $\Shares'$; they have the distribution
\begin{align}
P_{\Secrets',\Shares'}(\secrets,\shares) &= P_{\Secrets'}(\secrets)P_R(\{r\in\cR:\encoder(\secrets,r)=\shares\}).
\label{eq:switched-distr}
\end{align}
Then $P_{\Sab,\Sca}=P_{\Sab',\Sca'}$.
\end{claim}
Therefore, $H(\Sab|\Sca) = H(\Sca''|\Sab'')$ for any $P_{\Secrets''}\in\pins^1(P_{\Xa})$. Hence,
\begin{align*}
H(\Sca|\Sab) &= H(\Sca''|\Sab'')\\ &
  \geq H(\Xa'',\Xc''|\Xb'') + RI(\Xa'';\Xb''),
\end{align*}
where the inequality follows from \eqref{eq:bound-cond-entropy2}.
Thus, the second conditional entropy term in \eqref{eq:firstLB2} may be lower bounded as
\begin{align} &H(\Sca|\Sab) \notag\\
 &\qquad\geq\max_{P_{\Secrets''}\in\pins^1(P_{\Xa})} H(\Xa'',\Xc''|\Xb'') + RI(\Xa'';\Xb'').\label{eq:bound-cond-entropy2-strengthened} \end{align}
Similarly strengthening the lower bound on the first conditional entropy term in \eqref{eq:firstLB2} gives \eqref{eq:LB2}.
\end{proof}

\section{3SS Schemes for Secrets with Binary Alphabets}\label{sec:binary}
In this section, we characterize the randomness complexity of 3SS for the 255 distinct non-empty domains of secrets $\cS\subseteq\{0,1\}^3$.
As a first step, we drastically reduce the number of distinct sets for which randomness complexity needs to be characterized to 21.
For this, we observe that the randomness complexity of two sets $\cS,\cS'\subseteq\{0,1\}^3$ coincide if $\cS'$ can be obtained from $\cS$ using one of the following transformations:
In the sequel, for brevity, we denote $(x_1,x_2,x_3)\in\{0,1\}^3$ by simply $x_1x_2x_3$ whenever there is no room for confusion.
\begin{description}[style=unboxed,leftmargin=0cm]
    \item[Negation.] For some $i\in[3]$, $\cS'$ is obtained by negating coordinate $i$ of elements of $\cS$.
    e.g., for $i=1$, $x_1x_2x_3\in\cS'$ if and only if $\bar{x}_1x_2x_3\in\cS$, where $\bar{x}_1=1-x_1$.
    A scheme for $\cS$ can be used for $\cS'$ (and vice versa) if $\party{i}$ negates the output of the reconstruction function.
    \item[Swapping coordinates.] For distinct $i,j\in[3]$, $\cS'$ is obtained by exchanging coordinates $i$ and $j$ of elements in $\cS$.
    e.g., for $(i,j)=(1,2)$, $x_1x_2x_3\in\cS'$ if and only if $x_2x_1x_3\in\cS$.
    A scheme for $\cS$ can be used for $\cS'$ (and vice versa) by exchanging the reconstruction functions of parties $\party{i}$ and $\party{j}$. 
\end{description}

Table~\ref{tab:cubestable} in the Appendix 
classifies all non-empty $\cS\subseteq\{0,1\}^3$ into 21 families such that any set in a family can be transformed to another from the same the family using a sequence of negations and swapping of coordinates. Thus all members of a family have the same randomness complexity.

%
%
\newcommand{\Depth}{0.5}
\newcommand{\Height}{0.5}
\newcommand{\Width}{0.5}

\newcommand \drawcubea[1][1]%
{
    \coordinate (O) at (0,0,0);
    \coordinate (C) at (0,\Height*#1,0);
    \coordinate (G) at (0,\Height*#1,\Depth*#1);
    \coordinate (D) at (0,0,\Depth*#1);
    \coordinate (A) at (\Width*#1,0,0);
    \coordinate (B) at (\Width*#1,\Height*#1,0);
    \coordinate (F) at (\Width*#1,\Height*#1,\Depth*#1);
    \coordinate (E) at (\Width*#1,0,\Depth*#1);
    
    \draw (O) -- (A) -- (E) -- (D) -- cycle;
    \draw (C) -- (B) -- (F) -- (G) -- cycle;
    \draw (O) -- (C);
    \draw (A) -- (B);
    \draw (E) -- (F);
    \draw (D) -- (G);
}

\newcommand \drawballa[1]
{
    \shade[ ball color = red, opacity = 1.0] (#1) circle (0.1cm);
}

\newcommand{\drawFamilyOne}
{
    \begin{tabular}{c}
        \begin{tikzpicture}
            \drawcubea
            \shade[ball color=white, opacity=0.0] (C) circle(0.15cm);
            \foreach \ballcoord in {{O}}
            {
                \drawballa{\ballcoord}
            }
        \end{tikzpicture}
    \end{tabular}
}

\newcommand{\drawFamilyTwo}
{
    \begin{tabular}{c}
        \begin{tikzpicture}
            \drawcubea
            \shade[ball color=white, opacity=0.0] (C) circle(0.15cm);
            \foreach \ballcoord in {{O},{D}}
            {
                \drawballa{\ballcoord} 
            }
        \end{tikzpicture}
    \end{tabular}
}

\newcommand{\drawFamilyThree}
{
    \begin{tabular}{c}
        \begin{tikzpicture}
            \drawcubea
            \shade[ball color=white, opacity=0.0] (C) circle(0.15cm);
            \foreach \ballcoord in {{O},{G}}
            {
                \drawballa{\ballcoord}
            }
        \end{tikzpicture}
    \end{tabular}
}

\newcommand{\drawFamilyFour}
{
    \begin{tabular}{c}
        \begin{tikzpicture}
            \drawcubea
            \shade[ball color=white, opacity=0.0] (C) circle(0.15cm);
            \foreach \ballcoord in {{O},{F}}
            {
                \drawballa{\ballcoord}
                
            }
        \end{tikzpicture}
    \end{tabular}
}

\newcommand{\drawFamilyFive}
{
    \begin{tabular}{c}
        \begin{tikzpicture}
            \drawcubea
            \shade[ball color=white, opacity=0.0] (C) circle(0.15cm);
            \foreach \ballcoord in {{O},{A},{D}}
            {
                \drawballa{\ballcoord}
                
            }
        \end{tikzpicture}
    \end{tabular}
}

\newcommand{\drawFamilySix}
{
    \begin{tabular}{c}
        \begin{tikzpicture}
            \drawcubea
            \shade[ball color=white, opacity=0.0] (C) circle(0.15cm);
            \foreach \ballcoord in {{O},{D},{B}}
            {
                \drawballa{\ballcoord}
                
            }
        \end{tikzpicture}
    \end{tabular}
}

\newcommand{\drawFamilySeven}
{
    \begin{tabular}{c}
        \begin{tikzpicture}
            \drawcubea
            \shade[ball color=white, opacity=0.0] (C) circle(0.15cm);
            \foreach \ballcoord in {{O},{G},{E}}
            {
                \drawballa{\ballcoord}
                
            }
        \end{tikzpicture}
    \end{tabular}
}

\newcommand{\drawFamilyEight}
{
    \begin{tabular}{c}
        \begin{tikzpicture}
            \drawcubea
            \shade[ball color=white, opacity=0.0] (C) circle(0.15cm);
            \foreach \ballcoord in {{O},{A},{D},{E}}
            {
                \drawballa{\ballcoord}
                
            }
        \end{tikzpicture}
    \end{tabular}
}

\newcommand{\drawFamilyNine}
{
    \begin{tabular}{c}
        \begin{tikzpicture}
            \drawcubea
            \shade[ball color=white, opacity=0.0] (C) circle(0.15cm);
            \foreach \ballcoord in {{O},{D},{C},{A}}
            {
                \drawballa{\ballcoord}
                
            }
        \end{tikzpicture}
    \end{tabular}
}

\newcommand{\drawFamilyTen}
{
    \begin{tabular}{c}
        \begin{tikzpicture}
            \drawcubea
            \shade[ball color=white, opacity=0.0] (C) circle(0.15cm);
            \foreach \ballcoord in {{O},{D},{C},{E}}
            {
                \drawballa{\ballcoord}
            }
        \end{tikzpicture}
    \end{tabular}
}

\newcommand{\drawFamilyEleven}
{
    \begin{tabular}{c}
        \begin{tikzpicture}
            \drawcubea
            \shade[ball color=white, opacity=0.0] (C) circle(0.15cm);
            \foreach \ballcoord in {{O},{D},{C},{F}}
            {
                \drawballa{\ballcoord}
            }
        \end{tikzpicture}
    \end{tabular}
}

\newcommand{\drawFamilyTwelve}
{
    \begin{tabular}{c}
        \begin{tikzpicture}
            \drawcubea
            \shade[ball color=white, opacity=0.0] (C) circle(0.15cm);
            \foreach \ballcoord in {{O},{D},{B},{F}}
            {
                \drawballa{\ballcoord}
            }
        \end{tikzpicture}
    \end{tabular}
}

\newcommand{\drawFamilyThirteen}
{
    \begin{tabular}{c}
        \begin{tikzpicture}
            \drawcubea
            \shade[ball color=white, opacity=0.0] (C) circle(0.15cm);
            \foreach \ballcoord in {{O},{G},{E},{B}}
            {
                \drawballa{\ballcoord}
            }
        \end{tikzpicture}
    \end{tabular}
}

\newcommand{\drawFamilyFourteen}
{
    \begin{tabular}{c}
        \begin{tikzpicture}
            \drawcubea
            \shade[ball color=white, opacity=0.0] (C) circle(0.15cm);
            \foreach \ballcoord in {{O},{D},{C},{A},{B}}
            {
                \drawballa{\ballcoord}
            }
        \end{tikzpicture}
    \end{tabular}
}

\newcommand{\drawFamilyFifteen}
{
    \begin{tabular}{c}
        \begin{tikzpicture}
            \drawcubea
            \shade[ball color=white, opacity=0.0] (C) circle(0.15cm);
            \foreach \ballcoord in {{O},{D},{C},{A},{F}}
            {
                \drawballa{\ballcoord}
            }
        \end{tikzpicture}
    \end{tabular}
}

\newcommand{\drawFamilySixteen}
{
    \begin{tabular}{c}
        \begin{tikzpicture}
            \drawcubea
            \shade[ball color=white, opacity=0.0] (C) circle(0.15cm);
            \foreach \ballcoord in {{O},{D},{C},{E},{B}}
            {
                \drawballa{\ballcoord}
            }
        \end{tikzpicture}
    \end{tabular}
}

\newcommand{\drawFamilySeventeen}
{
    \begin{tabular}{c}
        \begin{tikzpicture}
            \drawcubea
            \shade[ball color=white, opacity=0.0] (C) circle(0.15cm);
            \foreach \ballcoord in {{O},{D},{C},{A},{B},{G}}
            {
                \drawballa{\ballcoord}
            }
        \end{tikzpicture}
    \end{tabular}
}

\newcommand{\drawFamilyEighteen}
{
    \begin{tabular}{c}
        \begin{tikzpicture}
            \drawcubea
            \shade[ball color=white, opacity=0.0] (C) circle(0.15cm);
            \foreach \ballcoord in {{O},{D},{C},{G},{E},{B}}
            {
                \drawballa{\ballcoord}
            }
        \end{tikzpicture}
    \end{tabular}
}

\newcommand{\drawFamilyNineteen}
{
    \begin{tabular}{c}
        \begin{tikzpicture}
            \drawcubea
            \shade[ball color=white, opacity=0.0] (C) circle(0.15cm);
            \foreach \ballcoord in {{O},{D},{C},{E},{B},{F}}
            {
                \drawballa{\ballcoord}
            }
        \end{tikzpicture}
    \end{tabular}
}

\newcommand{\drawFamilyTwenty}
{
    \begin{tabular}{c}
        \begin{tikzpicture}
            \drawcubea
            \shade[ball color=white, opacity=0.0] (C) circle(0.15cm);
            \foreach \ballcoord in {{O},{D},{C},{G},{A},{E},{B}}
            {
                \drawballa{\ballcoord}
            }
        \end{tikzpicture}
    \end{tabular}
}

\newcommand{\drawFamilyTwentyOne}
{
    \begin{tabular}{c}
        \begin{tikzpicture}
            \drawcubea
            \shade[ball color=white, opacity=0.0] (C) circle(0.15cm);
            \foreach \ballcoord in {{O},{D},{C},{G},{A},{E},{B},{F}}
            {
                \drawballa{\ballcoord}
            }
        \end{tikzpicture}
    \end{tabular}
}

\begin{table*}[]
    \begin{tabular}{|@{ }c@{ }|@{ }c@{ }|@{ }c@{ }|@{ }c@{ }|}
    \hline
    Family & \begin{tabular}[c]{p{4cm}c}
        Representative $\cS \subseteq \{0,1\}^3$\ shown as a subset of nodes in a cube that is labeled as here:& 
        \begin{tikzpicture}[baseline=10]
            \drawcubea[1.2]
            \node [above right,scale=0.6] at (O) {$000$};
            \node [above right,scale=0.6] at (A) {$100$};
            \node [right,scale=0.6] at (E) {$101$};
            \node [left,scale=0.6] at (D) {$001$};
            \node [left,scale=0.6] at (C) {$010$};
            \node [right,scale=0.6] at (B) {$110$};
            \node [right,scale=0.6] at (F) {$111$};
            \node [left,scale=0.6] at (G) {$011$};
        \end{tikzpicture}
    \end{tabular} & A randomness-optimal 3SS scheme for \cS   & $\rho$ \\ \hline \hline
    1        & $\drawFamilyOne$                                                      & $\Sab=\Sbc=\Sca=0$                                                                                &   \\ \cline{1-3}
    2        & $\drawFamilyThree$                                                    & $\Sab=\Sca=0$, and $\Sbc=\xb=\xc$                                                                 & 0 \\ \cline{1-3}
    3        & $\drawFamilyFour$                                                     & $\Sab=\Sbc=\Sca=\xa=\xb=\xc$                                                                      &   \\ \hline
    4 - 6    & $\drawFamilyTwo$ $\drawFamilySix$ $\drawFamilyTwelve$                 & $R \sim $ Bern(1/2), $\Sab = \Xa =\xb$, $\Sbc = R$, $\Sca = R \oplus \xc$                         & 1 \\ \hline
    7 - 8    & $\drawFamilySeven$ $\drawFamilyThirteen$                              & $R \sim $ Bern(1/2), $\Sab = R$, $\Sbc = R \oplus \xb$, $\Sca = R \oplus \xa$                     & 1 \\ \hline
    9 - 10   & $\drawFamilyFive$ $\drawFamilyEight$                                  & $R_1, R_3 \sim$ Bern(1/2), $\Sca=(R_1,R_3)$, $\Sab = \xa \oplus R_1$,  $\Sbc = \xc \oplus R_3$    & 2 \\ \hline
    11 - 13  & $\drawFamilyNine$ $\drawFamilyFifteen$ $\drawFamilyEleven$            & \begin{tabular}[c]{@{}c@{}}$(\alpha, \beta, \gamma)$ be a uniform permutation of $(0,1,2)$;\\
                                                                                               $(\Sab,\Sbc,\Sca)   = \begin{cases} 
                                                                                                                        (\alpha,\beta,\gamma) &\text{ if } \secrets = 000 \\
                                                                                                                        (\alpha,\alpha,\alpha) &\text{ if } \secrets = 111 \\
                                                                                                                        (\alpha,\beta,\alpha) &\text{ if } \secrets =  100\\
                                                                                                                        (\alpha,\alpha,\beta) &\text{ if } \secrets =  010\\
                                                                                                                        (\beta,\alpha,\alpha) &\text{ if } \secrets =  001 
                                                                                                                    \end{cases}$\end{tabular}                                             & $\log_2 6$ \\ \hline
    14 - 21  & \begin{tabular}[c]{@{}c@{}}$\drawFamilyTen$ $\drawFamilyFourteen$    $\drawFamilySeventeen$   $\drawFamilySixteen$ \\     
                                          $\drawFamilyEighteen$  $\drawFamilyNineteen$  $\drawFamilyTwenty$  $\drawFamilyTwentyOne$ \end{tabular} 
                                                                                     & \begin{tabular}[c]{@{}c@{}}$R_1,R_2,R_3 \sim$ Bern(1/2),\\
                                                                                        $\Sab = (\xa \oplus R_1 , R_2)$, $\Sbc = (\xa \oplus R_2, R_3)$, 
                                                                                        $\Sca = (\xa \oplus R_3, R_1)$ \end{tabular}                                                     & $3$ \\ \hline 
    \end{tabular}
    \caption{\sc 3SS schemes for binary secrets}
    \label{tab:my-table}
    \end{table*}

\subsection{Randomness Optimal 3SS Schemes}
\Cref{tab:my-table} shows 3SS schemes which achieve optimal randomness complexity for each of the 21 aforementioned families.
Several families in the table share the same optimal 3SS scheme, for example families 4-6, 7-8, etc.
This is often a consequence of the fact that an 3SS scheme for $\cS\subseteq\{0,1\}^3$ is also a 3SS scheme for any $\cS'\subseteq\cS$ since the correctness and privacy conditions for $\cS'$ are included in those for $\cS$.
We club such cases together to avoid repetition.
The schemes are easily verified to be perfectly correct (see below for Families 11-13), and hence, the reconstruction functions are implicit in the schemes.
We skip their description to avoid redundancy.

In most schemes, each secret $\secret{i}$ is shared by providing $R$ in $W_{ij}$ and $\secret{i}\oplus R$ in $W_{ki}$, where $R$ is a uniformly random bit independent of every other random variable.
This ensures that $\secret{i}$ can not be reconstructed only using $W_j$ (which does not include $W_{ki}$) or $W_k$ (which does not include $W_{ij}$); this ensures privacy against $\party{j}$ and $\party{k}$ as required in \eqref{eq:privacy-Pi}.

We now argue the privacy of two schemes that do not adhere to this pattern:

\begin{description}[leftmargin=0pt]
    \item[Families 7-8.] 
    The scheme is private against $\party{1}$ since their shares $(R,R\oplus x_1)$ reveals only $x_1$.
    It is also private against $\party{2}$ for a similar reason.
    Shares of $\party{3}$ is $(R\oplus x_1,R\oplus x_2)$, which reveals $x_1\oplus x_2$.
    But, since all secrets $(x_1,x_2,x_3)$ in both these families have $0$-parity, $x_1\oplus x_2=x_3$, ensuring privacy against $\party{3}$.

    \item[Families 11-13.] 
    In this scheme, if $x_i=1$, the pair of shares seen by $\party{i}$ have identical values and this is a (uniformly) random element from $\{0,1,2\}$. On the other hand, if $x_i=0$, the shares seen by $\party{i}$ have distinct values and these are uniformly random elements from $\{0,1,2\}$ subject to being distinct.
    This holds irrespective of the values of the other secrets.
    Thus, the scheme is perfectly private. The reconstruction function is also clear from the above -- The secret $x_i=1$ if the shares of $\party{i}$ are identical and $x_i=0$ otherwise.
\end{description}


\subsection{Matching Lower Bounds for Schemes in \Cref{tab:my-table}}
For all the families except~13 and~14 in \Cref{tab:my-table}, tight lower bounds are obtained using \eqref{eq:LB2} of \Cref{thm:RClowerbound} by appropriately choosing the distributions $P_{\Secrets},P_{\Secrets'}$ and $P_{\Secrets''}$.
We illustrate this using an example below. The remaining cases can be shown along similar lines (see Appendix~\ref{app:binaryITlowerbounds}).
\begin{example}\label{ex:it-lb}
Let $\cS = \{000,001,010,100\}$.
We use the bound \eqref{eq:LB2} from theorem \ref{thm:RClowerbound}.
For $\epsilon>0$, define $P_{\Secrets}$ as
\begin{align*}
    P_{\Secrets} (000) &= 1/3-\epsilon \quad P_{\Secrets} (001) = 1/3-\epsilon\\
    P_{\Secrets} (010) &= 1/3-\epsilon \quad P_{\Secrets} (100) = 3\epsilon
\end{align*}
Next, define $P_{\Secrets'}$ belonging to $\cP(P_{\Xb})$, i.e., $P_{\Secrets'}$ such that $P_{\Xb'}=P_{\Xb}$, as follows:
\begin{align*}
    P_{\Secrets'} (000) &= 1/3-\epsilon \quad &&P_{\Secrets'} (001) = 3\epsilon\\
    P_{\Secrets'}(010) &= 1/3-\epsilon \quad &&P_{\Secrets'}(100) = 1/3-\epsilon
\end{align*}
Define $P_{\Secrets'}$ and $P_{\Secrets''}$ belonging to family $\cP(P_{\Xa})$, as follows:
\begin{align*}
    P_{\Secrets''}(000) &= 1/2-2\epsilon \quad &&P_{\Secrets''}(001) = 1/2-2\epsilon\\
    P_{\Secrets''}(010) &= \epsilon \quad &&P_{\Secrets''}(100) = 3\epsilon
\end{align*}
With respect to $P_{\Secrets'}$, the ``characteristic graph'' of random variables ($\Xb',\Xc')$ is connected when $\epsilon>0$, where
characteristic graph is a bipartite graph over the supports of $\Xb'$ and $\Xc'$ such that edge $(x_2,x_3)$ exists if and only if $P_{\Xb'\Xc'}(x_2,x_3)>0$.
Hence, G\'acs-K\"orner common information $CI_{GK}(\Xb';\Xc')=0$ (see~\cite{GaKo73}), and, therefore, $RI(\Xb';\Xc')=I(\Xb';\Xc')$.
Furthermore, $RI(\Xa'';\Xb'')=I(\Xa'';\Xb'')$ with respect to $P_{\Secrets''}$ as well since the characteristic graph of ($\Xa'',\Xb'')$ is also connected.
Then,
\begin{align*}
    \rc
    &= H(\Xa',\Xb'|\Xc')+RI(\Xc';\Xb')\\
    &\hspace{2cm} + H(\Xa'',\Xc''|\Xb'')+RI(\Xa'';\Xb'')-H(\Xa)\\
    &=H(\Xa',\Xb'|\Xc')+H(\Xb')-H(\Xb'|\Xc')\\
    & + H(\Xa'',\Xc''|\Xb'')+H(\Xa'')-H(\Xa''|\Xb'') - H(\Xa)\\
    &=H(\Xb')+H(\Xa'|\Xb',\Xc')+H(\Xc''|\Xa'',\Xb'')
\end{align*}
Plugging in these values according to $P_{\Secrets'}$ and $P_{\Secrets''}$, and denoting the entropy of Bernoulli$(p)$ by $h(p)$,
\begin{align*}
    \rc
    &= h\left(1/3-\epsilon\right) + (2/3-2\epsilon)h(1/2) + (1-4\epsilon) h\left(1/2\right).
\end{align*}
As $\epsilon \rightarrow 0, \rc \rightarrow \log_2 6$.
Note that in the above argument, we made use of the fact that for $\epsilon>0$, the G\'acs-K\"orner information involved are 0 and then obtained the requisite bound in the limit as $\epsilon\rightarrow 0$. Using the limiting distributions directly in \eqref{eq:LB2} does not lead to a tight bound.  
\end{example}


For families 13 and 14, where the information-theoretic lower bounding technique in \Cref{thm:RClowerbound} falls short of being tight,
we give a combinatorial proof of the optimality of the schemes in \Cref{tab:my-table}.
Below, we sketch the combinatorial lower bound for family 13.
A proof for family 14 following the same approach is provided in Appendix~\ref{app:family14}.

\begin{example}\label{ex:family13}
    For support $\cS = \{000,001,010,111\}$, a lower bound of $\rc \geq \log 6$ can be proved using a combinatorial argument as follows.
    Let $M_{\secrets} = \{\shares \text{ s.t. } P_{\Shares|\Secrets}(\shares|\secrets)>0\}$.
    We prove $\rc \geq \log 6$ by showing that $|M_{\secrets}| \geq 6$ for some $\secrets \in \cS$.
    Let $\aone \in \salpab, \bone \in \salpbc, \cone \in \salpca$ such that $a_1b_1c_1\in M_{000}$.
    By privacy against $\pa$ -- \eqref{eq:privacy-Pi} for $i=3$,
    $$P_{\Sa|\Secrets}(\aone\cone|010) = P_{\Sa|\Secrets}(\aone\cone|000) > 0,$$
    where the inequality follows from $a_1b_1c_1\in M_{000}$.
    Hence, there exists $b_2$ such that $a_1b_2c_1\in M_{001}$.
    We can show that $b_2\neq b_1$ as follows:
    $\decoderc(b_1c_1)=0$ ensures correctness of $\party{3}$ since $a_1b_1c_1\in M_{000}$.
    But, since $a_1b_2c_1\in M_{001}$, by correctness of $\party{3}$, $\decoderb(b_2c_1)=1$.
    Hence, $b_1\neq b_2$.

    We proceed using similar arguments to add more entries to $M_{\secrets},\secrets\in\cS$ and to establish the properties of these entries.
    \begin{enumerate}[leftmargin=0.5cm] 
        \item By privacy against $\pa$, $a_1b_1c_1\in M_{000}\Rightarrow a_1b_3c_1\in M_{010}$ for some $b_3$.
        Since $a_1b_1c_1\in M_{000}$, $\decoderb(a_1b_1)=0$. But $\decoderb(a_1b_3)=1$ since $a_1b_3c_1\in M_{010}$. Hence, $b_3\neq b_1$.
        Similarly, since $a_1b_2c_1\in M_{001}$, $\decoderc(b_2c_1)=1$ but $\decoderc(b_3c_1)=0$; hence, $b_3\neq b_2$.
        Thus, $b_1\neq b_3\neq b_2$.
        \item By privacy against $\pc$, $a_1b_1c_1\in M_{000}\Rightarrow\atwo\bone\cone\in M_{010}$ for some $a_2$.
        Since $a_1b_1c_1\in M_{000}$, $\decoderb(\aone\bone)=0$ but $\decoderb(a_2b_1)=1$; hence $a_2\neq a_1$.
        \item By privacy against $\party{3}$, $a_1b_2c_1\in M_{001}\Rightarrow\athree\btwo\cone\in M_{111}$ for some $\athree$.
        Since $\decoderb(\aone\btwo)=0$ ($\because a_1b_2c_1\in M_{001}$) but $\decoderb(\athree\btwo)=1$, $\athree \neq \aone$; since $\decodera(\atwo\cone)=0$ ($\because \atwo\bone\cone\in M_{010}$), $\athree \neq \atwo$.
        \item By privacy against $\pb$ and 
        $a_3b_2c_1\in M_{111}\Rightarrow\athree\btwo\ctwo\in M_{010}$ for some $\ctwo$.
        Since $\decoderc(\btwo\cone)=1$ ($\because a_1b_2c_1\in M_{001}$) and $\decoderc(\btwo\ctwo)=0$, $\ctwo \neq \cone$.
        \item By privacy against $\pa$, $a_2b_1c_1\in M_{010}\Rightarrow a_2b_4c_1\in M_{001}$ and $a_3b_2c_2\in M_{010}\Rightarrow a_3b_{5}c_2\in M_{001}$ for some $b_4,b_5$.
        \item By privacy against $\pa$, $a_2b_1c_1\in M_{010}\Rightarrow a_2b'_4c_1\in M_{000}$ and $a_3b_2c_2\in M_{010}\Rightarrow a_3b'_5c_2\in M_{000}$ for some $b'_4,b'_5$.
        Since $\decoderc(\bfour\cone)=1$ $(\because a_2b_4c_1\in M_{001})$, $b'_4\neq b_4$.
        Since $\decoderc(\bfive\ctwo)=1$ $(\because a_3b_5c_2\in M_{001})$, $b'_5\neq b_5$.
        \item By privacy against $\pb$, $a_2b'_4c_1\in M_{000}\Rightarrow a_2b'_4c_3\in M_{001}$ and $a_3b'_5c_2\in M_{000}\Rightarrow a_3b'_5c_4\in M_{001}$ for some $c_3,c_4$.
        \item By privacy against $\pb$, $a_1b_1c_1\in M_{000}\Rightarrow a_1b_1c_6\in M_{001}$. Since $\decoderc(b_1c_1)=0 (\because a_1b_1c_1 \in M_{000})$ but $\decoder(b_1c_6)=1$, $b_1 \neq b_6$. 
    \end{enumerate}
    Thus, we have distinct entries
    \begin{align*}
        \aone\btwo\cone, \aone\bone\csix, \atwo\bfour\cone, \athree\bfive\ctwo, \atwo\bfour'\cthree, \athree\bfive'\cfour \in M_{001}
    \end{align*}
    Hence, $|M_{001}| \geq 6$. Therefore, $\rho\geq\log6$ which matches that of the scheme in Table~\ref{tab:my-table}. 
\end{example}

\section{Discussion}\label{sec:discussion}
In all but two families in Section~\ref{sec:binary}, the information-theoretic lower bound in Theorem~\ref{thm:RClowerbound} is tight. Thus, by Remark~\ref{rem:tensorize}, for these, we also solve the randomness complexity of 3SS for domain of secrets $\cS^n$ for all $n$ (something a combinatorial argument of the kind we used to solve the remaining two families cannot do). 
The remaining two families point to an interesting open question which is an analog of the ``direct-sum'' problem in communication complexity~\cite{JComp:FeKuNaNi95,IT:BrRa14}.
From Remark~\ref{rem:tensorize} it is clear that when the bound in Theorem~\ref{thm:RClowerbound} is tight for the one-shot problem, there is no reduction in $\frac{1}{n}\rc(\cS^n)$ possible by working over a block of length $n>1$. Thus, no such reduction is possible for 3SS with binary secrets except potentially for these two families. 
One possibility is that there are schemes which can take advantage of the block length to reduce the amortized randomness complexity in one or both of these cases (or some other example with non-binary secrets)\footnote{While we do not pursue it here, the lower bounds in Theorem~\ref{thm:RClowerbound} also hold for a slightly weaker form of security where privacy and correctness are only asymptotically achieved (i.e., the quantities in \eqref{eq:ITcorrect} and \eqref{eq:ITprivacy} only approach 0 in the limit $n\rightarrow \infty$ for all distributions $P_{\Secrets}\in\pins$ such that the $n$-fold secrets are distributed as $n$ i.i.d. copies of $P_{\Secrets}$, i.e., $P_{\Secrets^n}=P_{\Secrets}\times\ldots\times P_{\Secrets}$). The schemes under the possibility discussed may only achieve this weaker form of security.}. The other possibility is that no such advantage exists and the weakness is in the bound of Theorem~\ref{thm:RClowerbound} on the limiting value of $\frac{1}{n}\rc(\cS^n)$.

\balance


\balance
\begin{table*}[]
\centering
    \begin{tabular}{|@{ }c@{ }|@{ }c@{ }|@{ }c@{ }|@{ }c@{ }|}
    \hline
    Family & \begin{tabular}[c]{p{4cm}c}
        Representative $\cS \subseteq \{0,1\}^3$\ shown as a subset of nodes in a cube that is labeled as here:& 
        \begin{tikzpicture}[baseline=10]
            \drawcubea[1.2]
            \node [above right,scale=0.6] at (O) {$000$};
            \node [above right,scale=0.6] at (A) {$100$};
            \node [right,scale=0.6] at (E) {$101$};
            \node [left,scale=0.6] at (D) {$001$};
            \node [left,scale=0.6] at (C) {$010$};
            \node [right,scale=0.6] at (B) {$110$};
            \node [right,scale=0.6] at (F) {$111$};
            \node [left,scale=0.6] at (G) {$011$};
        \end{tikzpicture}
    \end{tabular} & \begin{tabular}[c]{c}Optimizing distributions in \eqref{eq:LB2}\\ (the limiting values are shown)\end{tabular}   & $\rho$ \\ \hline \hline
    4 - 6        & $\drawFamilyTwo$ $\drawFamilySix$ $\drawFamilyTwelve$                                                    &      $P_{\Secrets'}=P_{\Secrets''}=P_{\Secrets}: P_{\Xc}(0)=P_{\Xc}(1)=1/2$     & 1 \\ \cline{1-4}
    7 - 8        & $\drawFamilySeven$ $\drawFamilyThirteen$                                                     &      $P_{\Secrets'}=P_{\Secrets''}=P_{\Secrets}: P_{\Xb}(0) = P_{\Xb}(1) =1/2$  &  1 \\ \hline
    9 - 10   & $\drawFamilyFive$ $\drawFamilyEight$                                  &   \begin{tabular}[c]{@{}c@{}}$P_{\Secrets}: \Xa=\Xb=0, P_{\Xc}(0)=P_{\Xc}(1)=1/2 $\\ $P_{\Secrets'}: \Xc=\Xb=0, P_{\Xa}(0)=P_{\Xa}(1)=1/2$\\ $P_{\Secrets''}: \Xa=\Xb=0, P_{\Xc}(0)=P_{\Xc}(1)=1/2$\end{tabular}  & 2 \\ \hline
    11 - 12 & $\drawFamilyNine$ $\drawFamilyFifteen$ & \begin{tabular}[c]{@{}c@{}}$P_{\Secrets}(000)= P_{\Secrets}(001)= P_{\Secrets}(010)=1/3$\\ $P_{\Secrets'}(000)= P_{\Secrets'}(010)= P_{\Secrets'}(100)=1/3$\\ $P_{\Secrets''}(000)=P_{\Secrets''}(001)=1/2$\end{tabular} & $\log 6$ \\ \hline
    13 & $\drawFamilyEleven$ & combinatorial bound (Example~\ref{ex:family13}) & $\log 6$ \\ \hline
    14 & $\drawFamilyTen$ & combinatorial bound (Appendix~\ref{app:family14}) & 3 \\ \hline
    15 - 16 & $\drawFamilyFourteen$    $\drawFamilySeventeen$ & \begin{tabular}[c]{@{}c@{}}$ P_{\Secrets}(000)=P_{\Secrets}(010)=1/2$\\ $ P_{\Secrets'}(000)=P_{\Secrets'}(100)=P_{\Secrets'}(010)=P_{\Secrets'}(110)=1/4$\\ $P_{\Secrets''}(000)=P_{\Secrets''}(001)=1/2$\end{tabular} & 3 \\ \hline
    17 - 21 & \begin{tabular}[c]{@{}c@{}}  $\drawFamilySixteen$ $\drawFamilyEighteen$  $\drawFamilyNineteen$     
          $\drawFamilyTwenty$  $\drawFamilyTwentyOne$ \end{tabular}  &  \begin{tabular}[c]{@{}c@{}}$P_{\Secrets}(000)=1/2, P_{\Secrets}(010)=1/2$\\ $P_{\Secrets'}(001)=P_{\Secrets'}(010)=P_{\Secrets'}(101)=P_{\Secrets'}(110)=1/4$\\ $P_{\Secrets''}(000)=1/2, P_{\Secrets''}(001)=1/2$\end{tabular} & $3$ \\ \hline 
    \end{tabular}
    \caption{\sc Lower bounds from Theorem~\ref{thm:RClowerbound}}
    \label{tab:lb}
    \end{table*}

\appendices
\section{Details Omitted from the Proof of Theorem~\ref{thm:RClowerbound}}\label{app:lowerbound}

\begin{proof}[Proof of Claim~\ref{cl:useRI}]
Residual information (Definition~\ref{def:RI}) has the following property:
\begin{lemma}[Secure Data Processing Inequality~{\cite[Lemma~6]{IT:WoWu08}}]\label{lem:secDPI}
If $T,U,V,Z$ are jointly distributed random variables such that the following two conditions hold: $I(U;Z|T)=0$ and $I(T;V|Z)=0$, then 
\begin{align}
 RI((U,T);(V,Z)) &\geq RI(T;Z). \label{eq:RI-DPineq}
\end{align}
\end{lemma}
\noindent Here, $(U,T)$ and $(V,Z)$ should be thought of as the ``views'' of two parties and $T$ and $Z$ their respective secrects. The condition $I(U;Z|T)=I(T;V|Z)=0$ are privacy guarantees; e.g., $I(U;Z|T)=0$ means the ``view'' $(U,T)$ of a party contains no additional information about the other party's secret $Z$ than can be inferred from party's own secret $T$. The lemma states that if such privacy guarantees hold, the residual information of the views cannot be any smaller than the residual information of the secrets.

\begin{align*}
I(\Sab;\Sbc|\Sca) &= I(\Sab,\Sca;\Sbc,\Sca|\Sca)\\
        &\stackrel{(a)}{=} I(\Sab,\Sca,\Xa; \Sbc,\Sca,\Xc|\Sca)\\
        &\stackrel{(b)}{\geq} RI(\Sab,\Sca,\Xa;\Sbc,\Sca,\Xc)\\
        &\stackrel{(c)}{\geq} RI(\Xa;\Xc),
\end{align*}
where (a) follows from correctness conditions \eqref{eq:ITcorrect}, (b) from the definition of residual information \eqref{eq:RIdef}, and (c)  from Lemma~\ref{lem:secDPI} since the privacy conditions \eqref{eq:ITprivacy} imply that $I(\Sbc,\Sca;\Xa|\Xc)=I(\Sab,\Sca;\Xc|\Xa)=0$.
The other two inequalities \eqref{eq:useRI2} and \eqref{eq:useRI3} follow similarly.
\end{proof}

\begin{proof}[Proof of Claim~\ref{cl:switch-distr}]
As defined earlier, denote $\Sa=(\Sab,\Sca)$ and $\Sa'=(\Sab',\Sca')$. We need to show that $P_{\Sa}=P_{\Sa'}$. We will see that this is a direct consequence of privacy against $\pa$, i.e., \eqref{eq:privacy-Pi} with $i=1$. This privacy condition implies that
$P_{\Sa|\Secrets} (\ssa|\xa,\xb,\xc)=P_{\Sa'|\Secrets'}(\ssa|\xa,\xb,\xc)$ is only a function of $\ssa,\xa$. Let this function be $Q(\ssa|\xa)$. Then, using this we may write, for $\ssa\in\salpa$,
\begin{align*} 
P_{\Sa}(\ssa) &= \sum_{(\xa,\xb,\xc)\in\cS} P_{\Secrets}(\xa,\xb,\xc) P_{\Sa|\Secrets} (\ssa|\xa,\xb,\xc)\\
 &= \sum_{(\xa,\xb,\xc)\in\cS} P_{\Secrets}(\xa,\xb,\xc) Q(\ssa|\xa)\\
 &= \sum_{\xa} P_{\Xa}(\xa) Q(\ssa|\xa)\\
 &\stackrel{\text{(a)}}{=} \sum_{\xa} P_{\Xa'}(\xa) Q(\ssa|\xa)\\
 &= \sum_{(\xa,\xb,\xc)\in\cS} P_{\Secrets'}(\xa,\xb,\xc) Q(\ssa|\xa)\\
 &= \sum_{(\xa,\xb,\xc)\in\cS} P_{\Secrets'}(\xa,\xb,\xc) P_{\Sa'|\Secrets'} (\ssa|\xa,\xb,\xc)\\
 &= P_{\Sa'}(\ssa),
\end{align*}
where (a) follows from $P_{\Xa'}=P_{\Xa}$.
\end{proof}

\begin{proof}[Proof of \eqref{eq:LB1}]
Along the lines of the proof of Claim~\ref{cl:firstLB1} (where in the last step $H(\Sab)$ is lower bounded by $H(\Sab|\Sca)$),
\begin{align} \log|\cR| \geq H(\Sab|\Sca) + H(\Sca|\Sab) - H(\Xa).\label{eq:firstLB1}\end{align}
Along the lines of \eqref{eq:bound-cond-entropy3}, we may argue that
\begin{align}
H(\Sab|\Sca) &\geq H(\Xa,\Xb|\Xc) + RI(\Xa;\Xc).\label{eq:bound-cond-entropy1}
\end{align}
This can be further strengthened using Claim~\ref{cl:switch-distr} to
\[ H(\Sab|\Sca) \geq \max_{P_{\Secrets'}\in\pins^1(P_{\Xa})} H(\Xa',\Xb'|\Xc') + RI(\Xa';\Xc').\]
Now \eqref{eq:LB1} follows from applying to \eqref{eq:firstLB1}  this and the strengthened lower bound on $H(\Sca|\Sab)$ we obtained in the proof of~\eqref{eq:LB2}.
\end{proof}

\section{Tight Lower bounds for Binary 3SS}\label{app:binaryITlowerbounds}
The information-theoretic lower bound \eqref{eq:LB2} on randomness complexity of Theorem~\ref{thm:RClowerbound} turns out to be tight for all families of problems from Section~\ref{sec:binary} except two (families 13, 14). This can be shown along the lines of the discussion in Example~\ref{ex:it-lb}. Table~\ref{tab:lb} shows the optimizing distributions $P_{\Secrets},P_{\Secrets'},P_{\Secrets''}$ of \eqref{eq:LB2} which lead to the tight bounds. Note that, as in Example~\ref{ex:it-lb}, the distributions are taken to have full support over $\cS$ and the optimal bound is obtained in the limit as the distributions approach those shown in the table. Note that a lower bound shown for $\cS$ serves as a lower bound for any $\cS'$ such that $\cS'\supset \cS$. In the table we have grouped together such cases and it suffices to show the lower bound for the one whose domain of secrets is contained in all the others.

A combinatorial argument for family 13 was already shown in Example~\ref{ex:family13}; that for family~14 is given in Appendix~\ref{app:family14}.

    \begin{table}
        \centering
        \begin{tabular}{|@{ }c@{ }|@{ }c@{ }|@{ }c@{ }|@{ }c@{ }|}
            \hline
            $M_{000}$ & $M_{010}$ & $M_{100}$ & $M_{101}$\\ \hline
            $a_1b_1c_1$ & \begin{tabular}{@{ } c@{ }} $a_2b_1c_1$\\ $a_2 \neq a_1$ \end{tabular} & \begin{tabular}{@{ } c@{ }} $a_1b_1c_2$\\ $c_2 \neq c_1$ \end{tabular} & \begin{tabular}{@{ } c@{ }}$a_1b_1c_3$ \\ $c_1 \neq c_3 \neq c_2$ \end{tabular}\\ \hline
            \begin{tabular}{@{ } c@{ }} $a_4b_1c_2$\\ $a_2 \neq a_4 \neq a_1$\\ $a_4 \neq a_3$ \end{tabular}& \begin{tabular}{@{ } c@{ }} $a_1b_2c_1$ \\ $b_1 \neq b_2$ \end{tabular} & \begin{tabular}{@{ } c@{ }} $a_3b_1c_1$\\ $a_1 \neq a_3 \neq a_2$ \end{tabular} & \begin{tabular}{@{ } c@{ }} $a_3b_1c_4$\\ $c_4 \neq c_1$ \end{tabular}\\ \cline{1,2,3,4}
            \begin{tabular}{@{ } c@{ }} $a_2b_5c_1$\\$b_5 \neq b_1$ \end{tabular}& & \begin{tabular}{@{ } c@{ }} $a_1b_3c_3$ \\ $b_3 \neq b_1$\end{tabular}& \\ \cline{1,3}
            &  \begin{tabular}{@{ } c@{ }} \phantom{$a_2 \neq a_1 = a_1$} \end{tabular} & \begin{tabular}{@{ } c@{ }} $a_3b_6c_4$ \phantom{a} \end{tabular} & \begin{tabular}{@{ } c@{ }}\phantom{f}\\\phantom{f} \end{tabular} \\ \cline{3,4}
             &   & \begin{tabular}{@{ } c@{ }} $a_4b_1c_5$\\ \end{tabular} & \begin{tabular}{@{ } c@{ }} $a_4b_1c_5'$\\ $c_5 \neq c_5'$ \end{tabular}\\ \cline{3,4}
             &   &  \begin{tabular}{@{ } c@{ }} $a_4b_7c_5'$ \phantom{a} \end{tabular} & \begin{tabular}{@{ } c@{ }}\phantom{$a_4b_1c_5'$}\\ \phantom{$c_5 \neq c_5'$}\end{tabular}\\ \cline{3,4}
             &   &   \begin{tabular}{@{ } c@{ }} $a_2b_5c_6$\\ \end{tabular} & \begin{tabular}{@{ } c@{ }} $a_2b_5c_6'$\\ $c_6 \neq c_6'$ \end{tabular}\\ \cline{3,4}
             &   &  \begin{tabular}{@{ } c@{ }} $a_2b_8c_6'$ \phantom{a} \end{tabular} & \begin{tabular}{@{ } c@{ }}\phantom{$a_4b_1c_5'$}\\ \phantom{$c_5 \neq c_5'$}\end{tabular} \\ \hline
        \end{tabular}
        \caption{Table for Case-1}
        \label{tab:countcase1}
    \end{table}

\section{A Combinatorial Lower Bound for the Randomness Complexity of Family~14}\label{app:family14}

\begin{lemma}{\label{lemma:counting2}}
    Any scheme for 3SS on the support $\cS = \{000,010,100,101\}$ has $\rc \geq 3$.
\end{lemma}
\begin{proof}
    Let $M_{\Secrets} = \{\shares \text{ s.t. } P_{\Shares|\Secrets}(\shares|\secrets)>0\}$. We prove $\rc \geq 8$ by showing that $|M_{\secrets}| \geq 8$ for some $\secrets \in \cS$. Let $\salpab = \{a_1,a_2,\dots,a_{n_1}\}$, $\salpbc=\{b_1,b_2,\dots,b_{n_2}\}$, $\salpca = \{c_1,c_2,\dots,c_{n_3}\}$. We will denote a triple $(a_i,b_j,c_k)$ by $a_ib_jc_k$. Let some $a_1b_2c_3 \in \salp$ such that $\aone\bone\cone \in M_{000}$.

    Due to privacy against $\pb$ from \eqref{eq:privacy-Pi} for $i=3$, 
    $$P_{\Sb|\Secrets}(\aone\bone|000)=P_{\Sb|\Secrets}(\aone\bone|100) = P_{\Sb|\Secrets}(\aone\bone|101)$$ where this follows from the $\aone\bone\cone \in M_{000}$. Hence, there exist $\aone\bone\ctwo \in M_{100}$ and $\aone\bone\cthree \in M_{101}$. Note that $\cthree \neq \ctwo \neq \cone \neq \cthree$ by the following argument: For correctness at $\pa$, since, $\aone\bone\cone \in M_{000}$, $\decodera(\aone\cone)=0$. Also since $\aone\bone\ctwo \in M_{100}$ and $\aone\bone\cthree \in M_{101}$, for correctness at $\pa$, we must have $\decodera(\aone\ctwo)=1$ and $\decodera(\aone\cthree)=1$ respectively. This implies that $\cone \neq \ctwo$ and $\cone \neq \cthree$. Similarly, for correctness at $\pc$, due to $\aone\bone\ctwo \in M_{100}$ and $\aone\bone\cthree \in M_{101}$, $\decoderc(\bone\ctwo)=0$ and $\decoderc(\bone\cthree)=1$. This implies that $\ctwo\neq\cthree$. 
    We proceed using similar arguments and add entries to the sets $M_{\secrets}$ and observe that a scheme would necessarily need a necessary number of elements in at least one of these sets.

    \begin{enumerate}
        \item   Privacy against $\pc$ and $\aone\bone\cone \in M_{000}$ $\Rightarrow$ $\atwo\bone\cone \in M_{010}$ and $\athree\bone\cone \in M_{100}$ for some $\atwo$ and $\athree$. Since $\atwo\bone\cone \in M_{010}$ and $\athree\bone\cone \in M_{100}$, $\decoderb(\atwo\bone)=1$ and $\decoderb(\athree\bone)=0$. Also, since $\aone\bone\cone \in M_{000}$, $\decoder(\aone\bone)=0$ and $\decoder(\aone\cone)=0$. From all these, we have that $\aone \neq \atwo \neq \athree \neq \aone$. 

        \item   Privacy against $\pa$ and $\aone\bone\cone \in M_{000} \Rightarrow \aone\btwo\cone \in M_{010}$. Since $\decoderb(\aone\btwo)=1$ and $\decoderb(\aone\bone)=0$, $\bone \neq \btwo$.

        \item   Privacy against $\pc$ and $\aone\bone\ctwo \in M_{100} \Rightarrow a_i\bone\ctwo \in M_{000}$. Since $\decodera(\aone\ctwo)=1 (\because \aone\bone\ctwo \in M_{100}), a_i \neq \aone$. Since $\decoderb(\atwo\bone)=1 (\because \atwo\bone\cone \in M_{010})$, $a_i \neq \atwo$.
    \end{enumerate} 
    Going further, we divide the argument into two cases, one with $a_i \neq \athree$ and one with $a_i=\athree$ and continue adding the elements to the sets.  

    \textbf{Case 1 ($a_i \neq \athree$):} WLOG, assume that when $a_i = \afour$. So now, $\afour\bone\ctwo \in M_{000}$. In this case, $\aone,\atwo,\athree,\afour$, all are distinct.

    \begin{enumerate}
        \setcounter{enumi}{3}

        \item   Privacy against $\pa$ and $\aone\bone\cthree \in M_{101} \Rightarrow \aone\bthree\cthree \in M_{100}$ for some $\bthree$. Since $\decoder(\bone\cthree)=1 (\because \aone\bone\cthree \in M_{101})$ and $\decoder(\bthree\cthree)=0$, $\bthree \neq \bone$.
        
        \item   Privacy against $\pb$ and $\athree\bone\cone \in M_{100} \Rightarrow \athree\bone\cfour \in M_{101}$ for some $\cfour$. Since $\decoderc(\bone\cone) = 0 (\because \athree\bone\cone \in M_{100}$) and $\decoderc(\bone\cfour) = 1$, $\cfour \neq \cone$.

        \item   Privacy against $\pa$ and $\atwo\bone\cone \in M_{010} \Rightarrow \atwo\bfive\cone \in M_{000}$. Since $\decoderb(\atwo\bone) = 1 (\because \atwo\bone\cone \in M_{010})$ and $\decoderb(\atwo\bfive)=0$, $\bfive \neq \bone$. 
        
        \item   Privacy against $\pa$ and $\athree\bone\cfour \in M_{101} \Rightarrow \athree\bsix\cfour \in M_{100}$ for some $\bsix$.

        \item   Privacy against $\pb$ and $\afour\bone\ctwo \in M_{000}\\
        \Rightarrow \afour\bone\cfive \in M_{100}$ and $\afour\bone\cfive' \in M_{101}$ for some $c_5,c_5'$.\\$\Rightarrow \decoderc(\bone\cfive) = 0$ and $\decoderc(\bone\cfive')=1 \Rightarrow \cfive \neq \cfive'$.

        \item   Privacy against $\pa$ and $\afour\bone\cfive' \in M_{101} \Rightarrow \afour b_7 \cfive' \in M_{100}$ for some $b_7$.

        \item   Privacy against $\pb$ and $\atwo\bfive\cone \in M_{000}\\
        \Rightarrow \atwo\bfive\csix \in M_{100}$ and $\atwo\bfive\csix' \in M_{101}$ for some $\csix,\csix'$\\
        $\Rightarrow \decoderc(\bfive\csix)=0$ and $\decoderc(\bfive\csix')=1 \Rightarrow \csix \neq \csix'$.   

        \item   Privacy against $\pa$ and $\atwo\bfive\csix' \in M_{100} \Rightarrow \atwo b_8 \csix' \in M_{100}$ for some $b_8$. The argument is captured in table \ref{tab:countcase1}.
    \end{enumerate}

    Thus we have distinct entries 
    \begin{multline}
        \aone\bone\ctwo,\athree\bone\cone,\aone\bthree\cthree,\athree\bsix\cfour,\\\afour\bone\cfive,\afour b_7 \cfive', \atwo\bfive\csix,\atwo b_8 \csix' \in M_{100}
    \end{multline}
    This proves that in this case, we have $|M_{100}| \geq 8$.

    \textbf{Case 2 ($a_i = \athree$):}
    So, $\athree\bone\ctwo \in M_{000}$.
    \begin{enumerate}
        \setcounter{enumi}{3}
        \item   Privacy against $\pc$ and $\aone\bone\ctwo \in M_{100} \Rightarrow \afour\bone\ctwo \in M_{010}.$ Since $\decoderb(\aone\bone) = 0 (\because a_1b_1c_2 \in M_{100})$ and $\decoderb(\athree\bone)=0 (\because a_3b_1c_2 \in M_{000})$, $a_4 \neq a_3$ and $a_4 \neq a_1$. 
    \end{enumerate}

    \begin{table}
        \centering
        \begin{tabular}{|@{ }c@{ }|@{ }c@{ }|@{ }c@{ }|@{ }c@{ }|}
            \hline
            $M_{000}$ & $M_{010}$ & $M_{100}$ & $M_{101}$\\ \hline
            $a_1b_1c_1$ & \begin{tabular}{@{ } c@{ }} $a_2b_1c_1$\\ $a_2 \neq a_1$ \end{tabular} & \begin{tabular}{@{ } c@{ }} $a_1b_1c_2$\\ $c_2 \neq c_1$ \end{tabular} & \begin{tabular}{@{ } c@{ }}$a_1b_1c_3$ \\ $c_1 \neq c_3 \neq c_2$ \end{tabular}\\ \hline
            \begin{tabular}{@{ } c@{ }} $a_3b_1c_2$\\ $a_2 \neq a_3 \neq a_1$ \end{tabular}& \begin{tabular}{@{ } c@{ }} $a_1b_2c_1$ \\ $b_1 \neq b_2$ \end{tabular} & \begin{tabular}{@{ } c@{ }} $a_3b_1c_1$\\ $a_1 \neq a_3 \neq a_2$ \end{tabular} & \\ \hline
            \begin{tabular}{@{ } c@{ }} $a_2b_3c_1$ \end{tabular} & \begin{tabular}{@{ } c@{ }} $a_4b_1c_2$ \\ $a_3 \neq a_4 \neq a_2$\\ $a_4 \neq a_1$ \end{tabular} & \begin{tabular}{@{ } c@{ }} $a_2b_3c_4$ \end{tabular} & \begin{tabular}{@{ } c@{ }} $a_2b_3c_4'$\\$c_4 \neq c_4'$ \end{tabular}\\  \cline{1,2,3,4}
            \begin{tabular}{@{ } c@{ }} $a_4b_4c_2$ \end{tabular} & & \begin{tabular}{@{ } c@{ }} $a_2b_5c_4'$ \end{tabular} &\begin{tabular}{@{ } c@{ }} \phantom{$a_1b_7c_3$} \\ \phantom{$b_1 \neq b_7 \neq b_i$} \end{tabular}\\ \cline{1,3,4}
            & & \begin{tabular}{@{ } c@{ }} $a_4b_4c_5$ \end{tabular} & \begin{tabular}{@{ } c@{ }} $a_4b_4c_5'$\\ $c_4 \neq c_4'$ \end{tabular}\\ \cline{3,4}
            & & \begin{tabular}{@{ } c@{ }} $a_4b_6c_5'$ \end{tabular}& \begin{tabular}{@{ } c@{ }} \phantom{$a_1b_7c_3$} \\ \phantom{$b_1 \neq b_7 \neq b_i$} \end{tabular}\\ \cline{3,4}
            & & \begin{tabular}{@{ } c@{ }} $a_1b_7c_3$ \\ $b_1 \neq b_7$ \end{tabular} & \begin{tabular}{@{ } c@{ }} $a_3b_1c_6$ \end{tabular}\\ \cline{3,4}
            & \begin{tabular}{@{ } c@{ }} \phantom{$a_1b_7c_3$} \\ \phantom{$b_1 \neq b_7 \neq b_i$} \end{tabular} & \begin{tabular}{@{ } c@{ }} $a_3b_8c_6$\end{tabular} &\\\hline
        \end{tabular}
        \caption{Table for Case-2.1}
        \label{tab:countcase2}
    \end{table}

 We now analyze the two cases separately when $a_4 \neq a_2$ and when $a_4 = a_2$ in \textbf{Case 2.1} and \textbf{Case 2.2} respectively.
    
    \textbf{Case 2.1 ($a_4 \neq a_2$):} Note that due to previous arguments, $a_1,a_2,a_3$ and $a_4$ are all distinct here.
    
    \begin{enumerate}
        \setcounter{enumi}{4}
        \item   Privacy against $\pa$ and $\atwo\bone\cone \in M_{010} \Rightarrow \atwo\bthree\cone \in M_{000}$ for some $\bthree$.

        \item   Privacy against $\pa$ and $\afour\bone\ctwo \in M_{010}\ \Rightarrow \afour\bfour\ctwo \in M_{000}$ for some $\bfour$.
        
        \item   Privacy against $\pb$ and $\atwo\bthree\cone \in M_{000} \\
        \Rightarrow \atwo\bthree\cfour \in M_{100}$ and $\atwo\bthree\cfour' \in M_{101}$ for some $c_4,c_4'$\\ $\Rightarrow \decoderc(b_3c_4)=0$ and $\decoderc(b_3c_4')=1 \Rightarrow c_4 \neq c_4'$.

        \item   Privacy against $\pa$ and $a_2b_3c_4' \in M_{101} \Rightarrow a_2b_5c_4' \in M_{100}$ for some $b_5$.

        \item   Privacy against $\pb$ and $a_4b_4c_2 \in M_{000}$\\  $\Rightarrow a_4b_4c_5 \in M_{100}$ and $a_4 b_4 c_5' \in M_{101}$ for some $c_5,c_5'$\\ $\Rightarrow \decoderc(b_4c_5)=0$ and $\decoderc(b_4c_5')=1 \Rightarrow c_5 \neq c_5'$. 

        \item   Privacy against $\pa$ and $a_4b_4c_5' \in M_{101} \Rightarrow a_4b_6c_5' \in M_{100}$ for some $b_6$.

        \item   Privacy against $\pa$ and $a_1b_1c_3 \in M_{101} \Rightarrow a_1b_7c_3 \in M_{100}$. Since $\decoderc(b_1c_3)=1 (\because a_1b_1c_3 \in M_{101})$ and $\decoderc(b_7c_3)=0 (\because a_1b_7c_3 \in M_{100})$, $b_1 \neq b_7$.

        \item   Privacy against $\pb$ and $a_3b_1c_1 \in M_{100} \Rightarrow a_3b_1c_6 \in M_{101}$ for some $c_6$. Since $\decoderc(b_1c_1) = 0 (\because a_3b_1c_1 \in M_{100})$ and $\decoder(b_1c_6)=1$, $c_1 \neq c_6$.
        
        \item   Privacy against $\pa$ and $a_3b_1c_6 \in M_{101} \Rightarrow a_3b_8c_6 \in M_{100}$ for some $b_8$. 
    \end{enumerate}
    Thus we have distinct entries 
    \begin{multline}
        a_1b_1c_2, a_3b_1c_1, a_2b_3c_4, a_2b_5c_4',\\ a_4b_4c_5, a_4b_6c_5', a_1b_7c_3, a_3b_8c_6 \in M_{100}.
    \end{multline}
    This proves that in this case, we have $|M_{100}| \geq 8$. The argument is captured in table \ref{tab:countcase2}.

    \textbf{Case 2.2:}
    Since $a_4=a_2$, we now have that $a_2b_1c_2 \in M_{010}$. We will use the notation $a_2$ in place of $a_4$ from here onwards.

    \begin{enumerate}
        \setcounter{enumi}{4}
        \item   Privacy against $\pa$ and $a_1b_1c_2 \in M_{100} \Rightarrow a_1b_3c_2 \in M_{101}$. Since $\decoderc(b_1c_2)=0 (\because a_1b_1c_2 \in M_{100})$ but $\decoderc(b_3c_2)=1$, $b_1 \neq b_3$. Since $\decoderb(a_1b_2)=1 (\because a_1b_2c_1 \in M_{010})$ but $\decoderb(a_1b_3)=0$, $b_3 \neq b_2$.

        \item   Privacy against $\pc$ and $a_1b_2c_1 \in M_{010} \Rightarrow a_5b_2c_1 \in M_{100}$ for some $a_5$. Since $\decoderb(a_1b_2)=1 (\because a_1b_2c_1 \in M_{010})$ but $\decoderb(a_5b_2)=0, a_1 \neq a_5$. 
        
        $\decodera(a_2c_1)=0 (\because a_2b_1c_1 \in M_{010})$ but $\decodera(a_5c_1)=1$, $a_5 \neq a_2$. 
    \end{enumerate}

    \begin{table}[hb!]
        \centering
        \begin{tabular}{|@{ }c@{ }|@{ }c@{ }|@{ }c@{ }|@{ }c@{ }|}
            \hline
            $M_{000}$ & $M_{010}$ & $M_{100}$ & $M_{101}$\\ \hline
            $a_1b_1c_1$ & \begin{tabular}{@{ } c@{ }} $a_2b_1c_1$\\ $a_2 \neq a_1$ \end{tabular} & \begin{tabular}{@{ } c@{ }} $a_1b_1c_2$\\ $c_2 \neq c_1$ \end{tabular} & \begin{tabular}{@{ } c@{ }}$a_1b_1c_3$ \\ $c_1 \neq c_3 \neq c_2$ \end{tabular}\\ \hline
            \begin{tabular}{@{ } c@{ }} $a_3b_1c_2$ \end{tabular}& \begin{tabular}{@{ } c@{ }} $a_1b_2c_1$ \\ $b_1 \neq b_2$ \end{tabular} & \begin{tabular}{@{ } c@{ }} $a_3b_1c_1$\\ $a_1 \neq a_3 \neq a_2$ \end{tabular} & \begin{tabular}{@{ } c@{ }} $a_1b_3c_2$\\ $b_2 \neq b_3\neq b_1$ \end{tabular}\\ \hline
            \begin{tabular}{@{ } c@{ }} \end{tabular}  & $a_2b_1c_2$ & \begin{tabular}{@{ } c@{ }} $a_5b_2c_1$\\ $a_1 \neq a_5 \neq a_2$\\ $a_5 \neq a_3$ \end{tabular} & $a_3b_4c_1$ \\ \cline{2,3,4}
            \begin{tabular}{@{ } c@{ }} \phantom{$a_1b_7c_3$} \\ \phantom{$b_1 \neq b_7 \neq b_i$} \end{tabular} & & & \begin{tabular}{@{ } c@{ }} $a_3b_1c_4$\\ $c_4 \neq c_1$ \end{tabular}\\ \cline{1,4}
            \begin{tabular}{@{ } c@{ }} $a_5b_2c_5'$\\ $c_5 \neq c_5'$ \end{tabular}& & & $a_5b_2c_5$\\ \cline{1,4}
            $a_2b_5c_1$ & & & \begin{tabular}{@{ } c@{ }} $a_5b_6c_1$\\ $b_6 \neq b_2$ \end{tabular} \\ \cline{1,3,4}
            &  & $a_2b_5c_6$ & \begin{tabular}{@{ } c@{ }} $a_2b_5c_6'$\\ $c_6 \neq c_6'$\end{tabular}\\ \cline{3,4}
            \begin{tabular}{@{ } c@{ }} \phantom{$a_1b_7c_3$} \\ \phantom{$b_1 \neq b_7 \neq b_i$} \end{tabular}& \begin{tabular}{@{ } c@{ }} \phantom{$a_1b_7c_3$} \\ \phantom{$b_1 \neq b_7 \neq b_i$} \end{tabular}& & \begin{tabular}{@{ } c@{ }} $a_2b_7c_6$ \end{tabular} \\ \hline 
        \end{tabular}
        \caption{Table for Case-2.2.1}
        \label{tab:countcase3}
    \end{table}

    We now again split the analysis into two cases, \textbf{Case 2.2.1} where $a_5 \neq a_3$ and the other \textbf{Case 2.2.2} where $a_5 = a_3$

    \textbf{Case 2.2.1 ($a_5 \neq a_3$):}: So now, $a_5b_2c_1 \in M_{100}$. Note that due to previous arguments, $a_1,a_2,a_3$ and $a_5$ are all distinct.

    \begin{enumerate}
        \setcounter{enumi}{6}
        \item   Privacy against $\pa$ and $a_3b_1c_1 \in M_{100} \Rightarrow a_3b_4c_1 \in M_{101}$ for some $b_4$.

        \item   Privacy against $\pb$ and $a_3b_1c_1 \in M_{100} \Rightarrow a_3b_1c_4 \in M_{101}$ for some $c_4$. Since $\decoderc(b_1c_1)=0 (\because a_3b_1c_1 \in M_{100})$ but $\decoderc(b_1c_4)=1$, $c_1 \neq c_4$.

        \item   Privacy against $\pb$ and $a_5b_2c_1 \in M_{100}$\\ $\Rightarrow a_5b_2c_5 \in M_{101}$ and $a_5b_2c_5' \in M_{000}$ for some $c_5,c_5$'\\$\Rightarrow \decoderc(b_2c_5)=1$ and $\decoderc(b_2c_5')=0 \Rightarrow c_5 \neq c_5'$.

        \item   Privacy against $\pa$ and $a_5b_2c_1 \in M_{100} \Rightarrow a_5b_6c_1 \in M_{101}$ for some $b_6$. Since $\decoderc(b_2c_1)=0$ and $\decoder(b_6c_1)=1$, $b_2 \neq b_6$.

        \item   Privacy against $\pa$ and $a_2b_1c_1 \in M_{010} \Rightarrow a_2b_5c_1 \in M_{000}$ for some $b_5$. 
        
        \item   Privacy against $\pb$ and $a_2b_5c_1 \in M_{000}$\\ $\Rightarrow a_2b_5c_6 \in M_{100}$ and $a_2b_5c_6' \in M_{101}$ for some $c_6$, $c_6'$\\ $\Rightarrow \decoderc(b_5c_6)=0$ and $\decoderc(b_5c_6')=1 \Rightarrow c_6 \neq c_6'$.  

        \item   Privacy against $\pa$ and $a_2b_5c_6 \in M_{100} \Rightarrow a_2b_7c_6 \in M_{101}.$
    \end{enumerate}
    Thus we have distinct entries 
    \begin{multline}
        a_1b_1c_3, a_1b_3c_2, a_3b_4c_1, a_3b_1c_4,\\ a_5b_2c_5, a_5b_6c_1, a_2b_5c_6', a_2b_7c_6 \in M_{101} 
    \end{multline}
    So we have $|M_{101}| \geq 8$ in this case. The argument is captured in table \ref{tab:countcase3}

    \begin{table}[h!]
        \begin{tabular}{|@{ }c@{ }|@{ }c@{ }|@{ }c@{ }|@{ }c@{ }|}
            \hline
            $M_{000}$ & $M_{010}$ & $M_{100}$ & $M_{101}$\\ \hline
            $a_1b_1c_1$ & \begin{tabular}{@{ } c@{ }} $a_2b_1c_1$\\ $a_2 \neq a_1$ \end{tabular} & \begin{tabular}{@{ } c@{ }} $a_1b_1c_2$\\ $c_2 \neq c_1$ \end{tabular} & \begin{tabular}{@{ } c@{ }}$a_1b_1c_3$ \\ $c_1 \neq c_3 \neq c_2$ \end{tabular}\\ \hline
            \begin{tabular}{@{ } c@{ }} $a_3b_1c_2$ \end{tabular}& \begin{tabular}{@{ } c@{ }} $a_1b_2c_1$ \\ $b_1 \neq b_2$ \end{tabular} & \begin{tabular}{@{ } c@{ }} $a_3b_1c_1$\\ $a_1 \neq a_3 \neq a_2$ \end{tabular} & \begin{tabular}{@{ } c@{ }} $a_1b_3c_2$\\ $b_2 \neq b_3\neq b_1$ \end{tabular}\\ \hline
            \begin{tabular}{@{ } c@{ }} \phantom{$a_1b_7c_3$} \\ \phantom{$b_1 \neq b_7 \neq b_i$} \end{tabular}  & $a_2b_1c_2$ & \begin{tabular}{@{ } c@{ }} $a_3b_2c_1$ \end{tabular} &  \\ \cline{1,2,3}
            $a_6b_2c_1$& \begin{tabular}{@{ } c@{ }} $a_3b_4c_2$\\ $b_1 \neq b_4 \neq b_2$\\ $b_4 \neq b_3$ \end{tabular} & \begin{tabular}{@{ } c@{ }} $a_6b_2c_4$\\ $c_1 \neq c_4$ \end{tabular} & \\ \cline{1,2,3}
            \begin{tabular}{@{ } c@{ }} $a_7b_4c_2$ \end{tabular} & & \begin{tabular}{@{ } c@{ }} $a_7'b_4c_2$\\ $a_7 \neq a_7'$ \end{tabular} &\\ \cline{1,3}
            \begin{tabular}{@{ } c@{ }} \phantom{$a_1b_7c_3$} \\ \phantom{$b_1 \neq b_7 \neq b_i$} \end{tabular}& & \begin{tabular}{@{ } c@{ }} $a_7b_4c_5$ \end{tabular} & \\ \cline{1,3}
            \begin{tabular}{@{ } c@{ }} $a_1b_3c_6'$\\ $c_6 \neq c_6'$ \end{tabular}& & $a_1b_3c_6$ & \\ \cline{1,3}
            \begin{tabular}{@{ } c@{ }} \phantom{$a_1b_7c_3$} \\ \phantom{$b_1 \neq b_7 \neq b_i$} \end{tabular}& & \begin{tabular}{@{ } c@{ }} $a_8b_3c_6'$ \end{tabular} & \\ \hline
        \end{tabular}
        \caption{Table for Case-2.2.2}
        \label{tab:countcase4}
    \end{table}

    \textbf{Case 2.2.2 ($a_5 = a_3$)}: So, now, $a_3b_2c_1 \in M_{100}$. 

    \begin{enumerate}
        \setcounter{enumi}{6}
        \item   Privacy against $\pa$ and $a_3b_1c_2 \in M_{000} \Rightarrow a_3b_4c_2 \in M_{010}$ for some $b_4$. Since $\decoderb(a_3b_1)=0 (\because a_3b_1c_2 \in M_{000})$ but $\decoderb(a_3b_4)=1$, $b_4 \neq b_1$. Similarly, since $\decoderb(a_3b_2)=0 (\because a_3b_2c_1 \in M_{100}), b_4 \neq b_2$. Since $\decoderc(b_3c_2)=1 (\because a_1b_3c_2 \in M_{101})$ but $\decoderc(b_4,c_2)=0$, $b_4 \neq b_3$.

        Note that due to arguments in previous steps, we have $b_1,b_2,b_3$ and $b_4$ all distinct (refer to steps (2) and step (5) of Case 2.2).

        \item   Privacy against $\pc$ and $a_1b_2c_1 \in M_{010} \Rightarrow a_6b_2c_1 \in M_{000}$ for some $a_6$.

        \item   Privacy against $\pb$ and $a_6b_2c_1 \in M_{000} \Rightarrow a_6b_2c_4 \in M_{100}$. Since $\decodera(a_6c_1)=0 (\because a_6b_2c_1 \in M_{000})$ but $\decodera(a_6c_4)=1$, $c_1 \neq c_4$. 

        \item   Privacy against $\pc$ and $a_3b_4c_2 \in M_{010}$\\ $\Rightarrow a_7b_4c_2 \in M_{000}$ and $a_7'b_4c_2 \in M_{100}$ for some $a_7,a_7'$\\
        $\Rightarrow \decodera(a_7c_2) = 0$ and $\decodera(a_7'c_2)=1 \Rightarrow a_7 \neq a_7'$.

        \item   Privacy against $\pb$ and $a_7b_4c_2 \in M_{000} \Rightarrow a_7b_4c_5 \in M_{100}$ for some $c_5$.

        \item   Privacy against $\pb$ and $a_1b_3c_2 \in M_{101}$\\$\Rightarrow a_1b_3c_6 \in M_{100}$ and $a_1b_3c_6' \in M_{000}$\\$\Rightarrow \decodera (a_1c_6)=1$ and $\decodera(a_1c_6')=0 \Rightarrow c_6 \neq c_6'$.
        
        \item   Privacy against $\pc$ and $a_1b_3c_6' \in M_{000} \Rightarrow a_8b_3c_6' \in M_{100}$ for some $a_8$. 
    \end{enumerate}
    Thus we have distinct entries 
    \begin{multline}
        a_1b_1c_2, a_3b_1c_1, a_3b_2c_1,a_6b_2c_4,\\
        a_7'b_4c_2, a_7b_4c_5, a_1b_3c_6, a_8b_3c_6' \in M_{100} 
    \end{multline}
    This gives $|M_{100}| \geq 8$ for this case as well. Table \ref{tab:countcase4} captures this argument.

    Since we have $|M_{\secrets}| \geq 8$ for every case for some $\secrets \in \cS$, the randomness complexity $\log|\cR| \geq 3$. 
\end{proof}

\newcounter{end}

    

\newcommand \drawcubeb[2]
{
    \coordinate (O) at (#1,#2,0);
    \coordinate (C) at (#1,\Height+#2,0);
    \coordinate (G) at (#1,\Height+#2,\Depth);
    \coordinate (D) at (#1,#2,\Depth);
    \coordinate (A) at (\Width+#1,#2,0);
    \coordinate (B) at (\Width+#1,\Height+#2,0);
    \coordinate (F) at (\Width+#1,\Height+#2,\Depth);
    \coordinate (E) at (\Width+#1,#2,\Depth);
    
    \draw (O) -- (A) -- (E) -- (D) -- cycle;
    \draw (C) -- (B) -- (F) -- (G) -- cycle;
    \draw (O) -- (C);
    \draw (A) -- (B);
    \draw (E) -- (F);
    \draw (D) -- (G);
}

\newcommand* \familyOne 
{
{{{O}}, {{D}}, {{C}}, {{G}}, {{A}}, {{E}}, {{B}}, {{F}}},
}

\newcommand* \familyTwo 
{
{{{O}, {D}}, {{O}, {C}}, {{O}, {A}}, {{D}, {G}}, {{D}, {E}}, {{C}, {G}}, {{C}, {B}}, {{G}, {F}}},
{{{A}, {E}}, {{A}, {B}}, {{E}, {F}}, {{B}, {F}}},
}

\newcommand* \familyThree 
{
{{{O}, {G}}, {{O}, {E}}, {{O}, {B}}, {{D}, {C}}, {{D}, {A}}, {{D}, {F}}, {{C}, {A}}, {{C}, {F}}},
{{{G}, {E}}, {{G}, {B}}, {{A}, {F}}, {{E}, {B}}},
}

\newcommand* \familyFour
{
{{{O}, {F}}, {{D}, {B}}, {{C}, {E}}, {{G}, {A}}},
}

\newcommand* \familyFive
{
{{{O},{A},{B}},{{A},{B},{C}},{{B},{C},{O}},{{C},{O},{A}},{{G},{D},{E}},{{D},{E},{F}},{{E},{F},{G}},{{F},{G},{D}}},
{{{O},{A},{E}},{{E},{D},{O}},{{A},{E},{D}},{{D},{O},{A}},{{C},{B},{F}},{{B},{F},{G}},{{F},{G},{C}},{{G},{C},{B}}},
{{{A},{B},{F}},{{B},{F},{E}},{{F},{E},{A}},{{E},{A},{B}},{{C},{G},{D}},{{G},{D},{O}},{{D},{O},{C}},{{O},{C},{G}}},
}

\newcommand* \familySix
{
{{{O}, {D}, {B}}, {{O}, {D}, {F}}, {{O}, {C}, {E}}, {{O}, {C}, {F}}, {{O}, {G}, {A}}, {{O}, {G}, {F}}, {{O}, {A}, {F}}, {{O}, {E}, {F}}},
{{{O}, {B}, {F}}, {{D}, {C}, {E}}, {{D}, {C}, {B}}, {{D}, {G}, {A}}, {{D}, {G}, {B}}, {{D}, {A}, {B}}, {{D}, {E}, {B}}, {{D}, {B}, {F}}},
{{{C}, {G}, {A}}, {{C}, {G}, {E}}, {{C}, {A}, {E}}, {{C}, {E}, {B}}, {{C}, {E}, {F}}, {{G}, {A}, {E}}, {{G}, {A}, {B}}, {{G}, {A}, {F}}},
}

\newcommand* \familySeven
{
{{{O}, {G}, {E}}, {{O}, {G}, {B}}, {{O}, {E}, {B}}, {{D}, {C}, {A}}, {{D}, {C}, {F}}, {{D}, {A}, {F}}, {{C}, {A}, {F}}, {{G}, {E}, {B}}},
}

\newcommand* \familyEight
{
{{{O}, {D}, {C}, {G}}, {{O}, {D}, {A}, {E}}, {{O}, {C}, {A}, {B}}, {{D}, {G}, {E}, {F}}, {{C}, {G}, {B}, {F}}, {{A}, {E}, {B}, {F}}},
}

\newcommand* \familyNine
{
{{{O}, {D}, {C}, {A}}, {{O}, {D}, {G}, {E}}, {{O}, {C}, {G}, {B}}, {{O}, {A}, {E}, {B}}, {{D}, {C}, {G}, {F}}, {{D}, {A}, {E}, {F}}, {{C}, {A}, {B}, {F}}, {{G}, {E}, {B}, {F}}},
}

\newcommand* \familyTen
{
{{{O}, {D}, {C}, {E}}, {{O}, {D}, {C}, {B}}, {{O}, {D}, {G}, {A}}, {{O}, {D}, {G}, {F}}, {{O}, {D}, {A}, {B}}, {{O}, {D}, {E}, {F}}, {{O}, {C}, {G}, {A}}, {{O}, {C}, {G}, {F}}},
{{{O}, {C}, {A}, {E}}, {{O}, {C}, {B}, {F}}, {{O}, {A}, {E}, {F}}, {{O}, {A}, {B}, {F}}, {{D}, {C}, {G}, {E}}, {{D}, {C}, {G}, {B}}, {{D}, {G}, {A}, {E}}, {{D}, {G}, {B}, {F}}},
{{{D}, {A}, {E}, {B}}, {{D}, {E}, {B}, {F}}, {{C}, {G}, {A}, {B}}, {{C}, {G}, {E}, {F}}, {{C}, {A}, {E}, {B}}, {{C}, {E}, {B}, {F}}, {{G}, {A}, {E}, {F}}, {{G}, {A}, {B}, {F}}},
}

\newcommand* \familyEleven
{
{{{O}, {D}, {C}, {F}}, {{O}, {D}, {G}, {B}}, {{O}, {D}, {A}, {F}}, {{O}, {D}, {E}, {B}}, {{O}, {C}, {G}, {E}}, {{O}, {C}, {A}, {F}}, {{O}, {C}, {E}, {B}}, {{O}, {G}, {A}, {E}}},
{{{O}, {G}, {A}, {B}}, {{O}, {G}, {E}, {F}}, {{O}, {G}, {B}, {F}}, {{O}, {E}, {B}, {F}}, {{D}, {C}, {G}, {A}}, {{D}, {C}, {A}, {E}}, {{D}, {C}, {A}, {B}}, {{D}, {C}, {E}, {F}}},
{{{D}, {C}, {B}, {F}}, {{D}, {G}, {A}, {F}}, {{D}, {G}, {E}, {B}}, {{D}, {A}, {B}, {F}}, {{C}, {G}, {A}, {F}}, {{C}, {G}, {E}, {B}}, {{C}, {A}, {E}, {F}}, {{G}, {A}, {E}, {B}}},
}

\newcommand* \familyTwelve
{
{{{O}, {D}, {B}, {F}}, {{O}, {C}, {E}, {F}}, {{O}, {G}, {A}, {F}}, {{D}, {C}, {E}, {B}}, {{D}, {G}, {A}, {B}}, {{C}, {G}, {A}, {E}}},
}

\newcommand* \familyThirteen
{
{{{O}, {G}, {E}, {B}}, {{D}, {C}, {A}, {F}}},
}

\newcommand* \familyFourteen
{
{{{O}, {D}, {C}, {G}, {A}}, {{O}, {D}, {C}, {G}, {E}}, {{O}, {D}, {C}, {G}, {B}}, {{O}, {D}, {C}, {G}, {F}}, {{O}, {D}, {C}, {A}, {E}}, {{O}, {D}, {C}, {A}, {B}}, {{O}, {D}, {G}, {A}, {E}}, {{O}, {D}, {G}, {E}, {F}}},
{{{O}, {D}, {A}, {E}, {B}}, {{O}, {D}, {A}, {E}, {F}}, {{O}, {C}, {G}, {A}, {B}}, {{O}, {C}, {G}, {B}, {F}}, {{O}, {C}, {A}, {E}, {B}}, {{O}, {C}, {A}, {B}, {F}}, {{O}, {A}, {E}, {B}, {F}}, {{D}, {C}, {G}, {E}, {F}}},
{{{D}, {C}, {G}, {B}, {F}}, {{D}, {G}, {A}, {E}, {F}}, {{D}, {G}, {E}, {B}, {F}}, {{D}, {A}, {E}, {B}, {F}}, {{C}, {G}, {A}, {B}, {F}}, {{C}, {G}, {E}, {B}, {F}}, {{C}, {A}, {E}, {B}, {F}}, {{G}, {A}, {E}, {B}, {F}}},
}

\newcommand* \familyFifteen
{
{{{O}, {D}, {C}, {A}, {F}}, {{O}, {D}, {G}, {E}, {B}}, {{O}, {C}, {G}, {E}, {B}}, {{O}, {G}, {A}, {E}, {B}}, {{O}, {G}, {E}, {B}, {F}}, {{D}, {C}, {G}, {A}, {F}}, {{D}, {C}, {A}, {E}, {F}}, {{D}, {C}, {A}, {B}, {F}}},
}

\newcommand* \familySixteen
{
{{{O}, {D}, {C}, {E}, {B}}, {{O}, {D}, {C}, {E}, {F}}, {{O}, {D}, {C}, {B}, {F}}, {{O}, {D}, {G}, {A}, {B}}, {{O}, {D}, {G}, {A}, {F}}, {{O}, {D}, {G}, {B}, {F}}, {{O}, {D}, {A}, {B}, {F}}, {{O}, {D}, {E}, {B}, {F}}},
{{{O}, {C}, {G}, {A}, {E}}, {{O}, {C}, {G}, {A}, {F}}, {{O}, {C}, {G}, {E}, {F}}, {{O}, {C}, {A}, {E}, {F}}, {{O}, {C}, {E}, {B}, {F}}, {{O}, {G}, {A}, {E}, {F}}, {{O}, {G}, {A}, {B}, {F}}, {{D}, {C}, {G}, {A}, {E}}},
{{{D}, {C}, {G}, {A}, {B}}, {{D}, {C}, {G}, {E}, {B}}, {{D}, {C}, {A}, {E}, {B}}, {{D}, {C}, {E}, {B}, {F}}, {{D}, {G}, {A}, {E}, {B}}, {{D}, {G}, {A}, {B}, {F}}, {{C}, {G}, {A}, {E}, {B}}, {{C}, {G}, {A}, {E}, {F}}},
}

\newcommand* \familySeventeen
{
{{{O}, {D}, {C}, {G}, {A}, {E}}, {{O}, {D}, {C}, {G}, {A}, {B}}, {{O}, {D}, {C}, {G}, {E}, {F}}, {{O}, {D}, {C}, {G}, {B}, {F}}, {{O}, {D}, {C}, {A}, {E}, {B}}, {{O}, {D}, {G}, {A}, {E}, {F}}, {{O}, {D}, {A}, {E}, {B}, {F}}, {{O}, {C}, {G}, {A}, {B}, {F}}},
{{{O}, {C}, {A}, {E}, {B}, {F}}, {{D}, {C}, {G}, {E}, {B}, {F}}, {{D}, {G}, {A}, {E}, {B}, {F}}, {{C}, {G}, {A}, {E}, {B}, {F}}},
}

\newcommand* \familyEighteen
{
{{{O}, {D}, {C}, {G}, {A}, {F}}, {{O}, {D}, {C}, {G}, {E}, {B}}, {{O}, {D}, {C}, {A}, {E}, {F}}, {{O}, {D}, {C}, {A}, {B}, {F}}, {{O}, {D}, {G}, {A}, {E}, {B}}, {{O}, {D}, {G}, {E}, {B}, {F}}, {{O}, {C}, {G}, {A}, {E}, {B}}, {{O}, {C}, {G}, {E}, {B}, {F}}},
{{{O}, {G}, {A}, {E}, {B}, {F}}, {{D}, {C}, {G}, {A}, {E}, {F}}, {{D}, {C}, {G}, {A}, {B}, {F}}, {{D}, {C}, {A}, {E}, {B}, {F}}},
}

\newcommand* \familyNineteen
{
{{{O}, {D}, {C}, {E}, {B}, {F}}, {{O}, {D}, {G}, {A}, {B}, {F}}, {{O}, {C}, {G}, {A}, {E}, {F}}, {{D}, {C}, {G}, {A}, {E}, {B}}},
}

\newcommand* \familyTwenty
{
{{{O}, {D}, {C}, {G}, {A}, {E}, {B}}, {{O}, {D}, {C}, {G}, {A}, {E}, {F}}, {{O}, {D}, {C}, {G}, {A}, {B}, {F}}, {{O}, {D}, {C}, {G}, {E}, {B}, {F}}, {{O}, {D}, {C}, {A}, {E}, {B}, {F}}, {{O}, {D}, {G}, {A}, {E}, {B}, {F}}, {{O}, {C}, {G}, {A}, {E}, {B}, {F}}, {{D}, {C}, {G}, {A}, {E}, {B}, {F}}},
}

\newcommand* \familyTwentyOne
{
{{{O}, {D}, {C}, {G}, {A}, {E}, {B}, {F}}},
}


\newcommand* \drawFamiliesOne
{
\begin{tabular}{c} \begin{tikzpicture}[scale=0.8, every node/.style={scale=0.8}]
    \foreach \cubecountx in {0,1,2,3,4,5,6,7}
    {
        \foreach \cubecounty in {0}
        {
        \drawcubeb{\cubecountx}{\cubecounty}
            \foreach \row [count=\rowcount, evaluate=\rowcount as \rowcountm using int(\rowcount-1)] in \familyOne
            {
                \foreach \column [count=\columncount, evaluate=\columncount as \columncountm using int(\columncount-1)] in \row
                {
                    \ifnum \cubecountx = \columncountm
                    {
                        \ifnum \cubecounty = \rowcountm
                        {
                            \foreach \ballcoord in \column
                            {
                                \drawballa{\ballcoord}
                            }
                        }
                        \fi
                    }
                    \fi
                }
            }
        }
    }
\end{tikzpicture} \end{tabular}
}

\newcommand* \drawFamiliesTwo
{
\begin{tabular}{c} \begin{tikzpicture}[scale=0.8, every node/.style={scale=0.8}]
    \foreach \cubecountx in {0,1,2,3,4,5,6,7}
    {
        \foreach \cubecounty in {0,1}
        {
        \setcounter{end}{\cubecountx+10*\cubecounty};
        \ifnum \value{end} < 14
        {
        \drawcubeb{\cubecountx}{\cubecounty}
            \foreach \row [count=\rowcount, evaluate=\rowcount as \rowcountm using int(\rowcount-1)] in \familyTwo
            {
                \foreach \column [count=\columncount, evaluate=\columncount as \columncountm using int(\columncount-1)] in \row
                {
                    \ifnum \cubecountx = \columncountm
                    {
                        \ifnum \cubecounty = \rowcountm
                        {
                            \foreach \ballcoord in \column
                            {
                                \drawballa{\ballcoord}
                            }
                        }
                        \fi
                    }
                    \fi
                }
            }
        }
        \fi
        }
    }
\end{tikzpicture} \end{tabular}
}

\newcommand* \drawFamiliesThree
{
\begin{tabular}{c} \begin{tikzpicture}[scale=0.8, every node/.style={scale=0.8}]
    \foreach \cubecountx in {0,1,2,3,4,5,6,7}
    {
        \foreach \cubecounty in {0,1}
        {
        \setcounter{end}{\cubecountx+10*\cubecounty};
        \ifnum \value{end} < 14
        {
        \drawcubeb{\cubecountx}{\cubecounty}
            \foreach \row [count=\rowcount, evaluate=\rowcount as \rowcountm using int(\rowcount-1)] in \familyThree
            {
                \foreach \column [count=\columncount, evaluate=\columncount as \columncountm using int(\columncount-1)] in \row
                {
                    \ifnum \cubecountx = \columncountm
                    {
                        \ifnum \cubecounty = \rowcountm
                        {
                            \foreach \ballcoord in \column
                            {
                                \drawballa{\ballcoord}
                            }
                        }
                        \fi
                    }
                    \fi
                }
            }
        }
        \fi
        }
    }
\end{tikzpicture} \end{tabular}
}

\newcommand* \drawFamiliesFour{
\begin{tabular}{c} \begin{tikzpicture}[scale=0.8, every node/.style={scale=0.8}]
    \foreach \cubecountx in {0,1,2,3,4,5,6,7}
    {
        \foreach \cubecounty in {0}
        {
        \setcounter{end}{\cubecountx};
        \ifnum \value{end} < 4
        {
        \drawcubeb{\cubecountx}{\cubecounty}
            \foreach \row [count=\rowcount, evaluate=\rowcount as \rowcountm using int(\rowcount-1)] in \familyFour
            {
                \foreach \column [count=\columncount, evaluate=\columncount as \columncountm using int(\columncount-1)] in \row
                {
                    \ifnum \cubecountx = \columncountm
                    {
                        \ifnum \cubecounty = \rowcountm
                        {
                            \foreach \ballcoord in \column
                            {
                                \drawballa{\ballcoord}
                            }
                        }
                        \fi
                    }
                    \fi
                }
            }
        }
        \fi
        }
    }
\end{tikzpicture} \end{tabular}}

\newcommand* \drawFamiliesFive{
\begin{tabular}{c} \begin{tikzpicture}[scale=0.8, every node/.style={scale=0.8}]
    \foreach \cubecountx in {0,1,2,3,4,5,6,7}
    {
        \foreach \cubecounty in {0,1,2}
        {
        \drawcubeb{\cubecountx}{\cubecounty}
            \foreach \row [count=\rowcount, evaluate=\rowcount as \rowcountm using int(\rowcount-1)] in \familyFive
            {
                \foreach \column [count=\columncount, evaluate=\columncount as \columncountm using int(\columncount-1)] in \row
                {
                    \ifnum \cubecountx = \columncountm
                    {
                        \ifnum \cubecounty = \rowcountm
                        {
                            \foreach \ballcoord in \column
                            {
                                \drawballa{\ballcoord}
                            }
                        }
                        \fi
                    }
                    \fi
                }
            }
        }
    }
\end{tikzpicture} \end{tabular}}


\newcommand* \drawFamiliesSix{
\begin{tabular}{c} \begin{tikzpicture}[scale=0.8, every node/.style={scale=0.8}]
    \foreach \cubecountx in {0,1,2,3,4,5,6,7}
    {
        \foreach \cubecounty in {0,1,2}
        {
        \drawcubeb{\cubecountx}{\cubecounty}
            \foreach \row [count=\rowcount, evaluate=\rowcount as \rowcountm using int(\rowcount-1)] in \familySix
            {
                \foreach \column [count=\columncount, evaluate=\columncount as \columncountm using int(\columncount-1)] in \row
                {
                    \ifnum \cubecountx = \columncountm
                    {
                        \ifnum \cubecounty = \rowcountm
                        {
                            \foreach \ballcoord in \column
                            {
                                \drawballa{\ballcoord}
                            }
                        }
                        \fi
                    }
                    \fi
                }
            }
        }
    }
\end{tikzpicture} \end{tabular}}

\newcommand* \drawFamiliesSeven{
\begin{tabular}{c} \begin{tikzpicture}[scale=0.8, every node/.style={scale=0.8}]
    \foreach \cubecountx in {0,1,2,3,4,5,6,7}
    {
        \foreach \cubecounty in {0}
        {
        \drawcubeb{\cubecountx}{\cubecounty}
            \foreach \row [count=\rowcount, evaluate=\rowcount as \rowcountm using int(\rowcount-1)] in \familySeven
            {
                \foreach \column [count=\columncount, evaluate=\columncount as \columncountm using int(\columncount-1)] in \row
                {
                    \ifnum \cubecountx = \columncountm
                    {
                        \ifnum \cubecounty = \rowcountm
                        {
                            \foreach \ballcoord in \column
                            {
                                \drawballa{\ballcoord}
                            }
                        }
                        \fi
                    }
                    \fi
                }
            }
        }
    }
\end{tikzpicture} \end{tabular}}


\newcommand* \drawFamiliesEight{
\begin{tabular}{c} \begin{tikzpicture}[scale=0.8, every node/.style={scale=0.8}]
    \foreach \cubecountx in {0,1,2,3,4,5,6,7}
    {
        \foreach \cubecounty in {0}
        {
        \setcounter{end}{\cubecountx};
        \ifnum \value{end} < 6
        {
        \drawcubeb{\cubecountx}{\cubecounty}
            \foreach \row [count=\rowcount, evaluate=\rowcount as \rowcountm using int(\rowcount-1)] in \familyEight
            {
                \foreach \column [count=\columncount, evaluate=\columncount as \columncountm using int(\columncount-1)] in \row
                {
                    \ifnum \cubecountx = \columncountm
                    {
                        \ifnum \cubecounty = \rowcountm
                        {
                            \foreach \ballcoord in \column
                            {
                                \drawballa{\ballcoord}
                            }
                        }
                        \fi
                    }
                    \fi
                }
            }
        }
        \fi
        }
    }
\end{tikzpicture} \end{tabular}}


\newcommand* \drawFamiliesNine{
\begin{tabular}{c} \begin{tikzpicture}[scale=0.8, every node/.style={scale=0.8}]
    \foreach \cubecountx in {0,1,2,3,4,5,6,7}
    {
        \foreach \cubecounty in {0}
        {
        \drawcubeb{\cubecountx}{\cubecounty}
            \foreach \row [count=\rowcount, evaluate=\rowcount as \rowcountm using int(\rowcount-1)] in \familyNine
            {
                \foreach \column [count=\columncount, evaluate=\columncount as \columncountm using int(\columncount-1)] in \row
                {
                    \ifnum \cubecountx = \columncountm
                    {
                        \ifnum \cubecounty = \rowcountm
                        {
                            \foreach \ballcoord in \column
                            {
                                \drawballa{\ballcoord}
                            }
                        }
                        \fi
                    }
                    \fi
                }
            }
        }
    }
\end{tikzpicture} \end{tabular}}

\newcommand* \drawFamiliesTen{
\begin{tabular}{c} \begin{tikzpicture}[scale=0.8, every node/.style={scale=0.8}]
    \foreach \cubecountx in {0,1,2,3,4,5,6,7}
    {
        \foreach \cubecounty in {0,1,2}
        {
        \drawcubeb{\cubecountx}{\cubecounty}
            \foreach \row [count=\rowcount, evaluate=\rowcount as \rowcountm using int(\rowcount-1)] in \familyTen
            {
                \foreach \column [count=\columncount, evaluate=\columncount as \columncountm using int(\columncount-1)] in \row
                {
                    \ifnum \cubecountx = \columncountm
                    {
                        \ifnum \cubecounty = \rowcountm
                        {
                            \foreach \ballcoord in \column
                            {
                                \drawballa{\ballcoord}
                            }
                        }
                        \fi
                    }
                    \fi
                }
            }
        }
    }
\end{tikzpicture} \end{tabular}}

\newcommand* \drawFamiliesEleven{
\begin{tabular}{c} \begin{tikzpicture}[scale=0.8, every node/.style={scale=0.8}]
    \foreach \cubecountx in {0,1,2,3,4,5,6,7}
    {
        \foreach \cubecounty in {0,1,2}
        {
        \drawcubeb{\cubecountx}{\cubecounty}
            \foreach \row [count=\rowcount, evaluate=\rowcount as \rowcountm using int(\rowcount-1)] in \familyEleven
            {
                \foreach \column [count=\columncount, evaluate=\columncount as \columncountm using int(\columncount-1)] in \row
                {
                    \ifnum \cubecountx = \columncountm
                    {
                        \ifnum \cubecounty = \rowcountm
                        {
                            \foreach \ballcoord in \column
                            {
                                \drawballa{\ballcoord}
                            }
                        }
                        \fi
                    }
                    \fi
                }
            }
        }
    }
\end{tikzpicture} \end{tabular}}

\newcommand* \drawFamiliesTwelve{
\begin{tabular}{c} \begin{tikzpicture}[scale=0.8, every node/.style={scale=0.8}]
    \foreach \cubecountx in {0,1,2,3,4,5,6,7}
    {
        \foreach \cubecounty in {0}
        {
        \setcounter{end}{\cubecountx};
        \ifnum \value{end} < 6
        {
        \drawcubeb{\cubecountx}{\cubecounty}
            \foreach \row [count=\rowcount, evaluate=\rowcount as \rowcountm using int(\rowcount-1)] in \familyTwelve
            {
                \foreach \column [count=\columncount, evaluate=\columncount as \columncountm using int(\columncount-1)] in \row
                {
                    \ifnum \cubecountx = \columncountm
                    {
                        \ifnum \cubecounty = \rowcountm
                        {
                            \foreach \ballcoord in \column
                            {
                                \drawballa{\ballcoord}
                            }
                        }
                        \fi
                    }
                    \fi
                }
            }
        }
        \fi
        }
    }
\end{tikzpicture} \end{tabular}}

\newcommand* \drawFamiliesThirteen{
\begin{tabular}{c} \begin{tikzpicture}[scale=0.8, every node/.style={scale=0.8}]
    \foreach \cubecountx in {0,1,2,3,4,5,6,7}
    {
        \foreach \cubecounty in {0}
        {
        \setcounter{end}{\cubecountx};
        \ifnum \value{end} < 2
        {
        \drawcubeb{\cubecountx}{\cubecounty}
            \foreach \row [count=\rowcount, evaluate=\rowcount as \rowcountm using int(\rowcount-1)] in \familyThirteen
            {
                \foreach \column [count=\columncount, evaluate=\columncount as \columncountm using int(\columncount-1)] in \row
                {
                    \ifnum \cubecountx = \columncountm
                    {
                        \ifnum \cubecounty = \rowcountm
                        {
                            \foreach \ballcoord in \column
                            {
                                \drawballa{\ballcoord}
                            }
                        }
                        \fi
                    }
                    \fi
                }
            }
        }
        \fi
        }
    }
\end{tikzpicture} \end{tabular}}

\newcommand* \drawFamiliesFourteen{
\begin{tabular}{c} \begin{tikzpicture}[scale=0.8, every node/.style={scale=0.8}]
    \foreach \cubecountx in {0,1,2,3,4,5,6,7}
    {
        \foreach \cubecounty in {0,1,2}
        {
        \drawcubeb{\cubecountx}{\cubecounty}
            \foreach \row [count=\rowcount, evaluate=\rowcount as \rowcountm using int(\rowcount-1)] in \familyFourteen
            {
                \foreach \column [count=\columncount, evaluate=\columncount as \columncountm using int(\columncount-1)] in \row
                {
                    \ifnum \cubecountx = \columncountm
                    {
                        \ifnum \cubecounty = \rowcountm
                        {
                            \foreach \ballcoord in \column
                            {
                                \drawballa{\ballcoord}
                            }
                        }
                        \fi
                    }
                    \fi
                }
            }
        }
    }
\end{tikzpicture} \end{tabular}}

\newcommand* \drawFamiliesFifteen{
\begin{tabular}{c} \begin{tikzpicture}[scale=0.8, every node/.style={scale=0.8}]
    \foreach \cubecountx in {0,1,2,3,4,5,6,7}
    {
        \foreach \cubecounty in {0}
        {
        \drawcubeb{\cubecountx}{\cubecounty}
            \foreach \row [count=\rowcount, evaluate=\rowcount as \rowcountm using int(\rowcount-1)] in \familyFifteen
            {
                \foreach \column [count=\columncount, evaluate=\columncount as \columncountm using int(\columncount-1)] in \row
                {
                    \ifnum \cubecountx = \columncountm
                    {
                        \ifnum \cubecounty = \rowcountm
                        {
                            \foreach \ballcoord in \column
                            {
                                \drawballa{\ballcoord}
                            }
                        }
                        \fi
                    }
                    \fi
                }
            }
        }
    }
\end{tikzpicture} \end{tabular}}

\newcommand* \drawFamiliesSixteen{
\begin{tabular}{c} \begin{tikzpicture}[scale=0.8, every node/.style={scale=0.8}]
    \foreach \cubecountx in {0,1,2,3,4,5,6,7}
    {
        \foreach \cubecounty in {0,1,2}
        {
        \drawcubeb{\cubecountx}{\cubecounty}
            \foreach \row [count=\rowcount, evaluate=\rowcount as \rowcountm using int(\rowcount-1)] in \familySixteen
            {
                \foreach \column [count=\columncount, evaluate=\columncount as \columncountm using int(\columncount-1)] in \row
                {
                    \ifnum \cubecountx = \columncountm
                    {
                        \ifnum \cubecounty = \rowcountm
                        {
                            \foreach \ballcoord in \column
                            {
                                \drawballa{\ballcoord}
                            }
                        }
                        \fi
                    }
                    \fi
                }
            }
        }
    }
\end{tikzpicture} \end{tabular}}

\newcommand* \drawFamiliesSeventeen{
\begin{tabular}{c} \begin{tikzpicture}[scale=0.8, every node/.style={scale=0.8}]
    \foreach \cubecountx in {0,1,2,3,4,5,6,7}
    {
        \foreach \cubecounty in {0,1}
        {
        \setcounter{end}{\cubecountx+10*\cubecounty};
        \ifnum \value{end} < 14
        {
        \drawcubeb{\cubecountx}{\cubecounty}
            \foreach \row [count=\rowcount, evaluate=\rowcount as \rowcountm using int(\rowcount-1)] in \familySeventeen
            {
                \foreach \column [count=\columncount, evaluate=\columncount as \columncountm using int(\columncount-1)] in \row
                {
                    \ifnum \cubecountx = \columncountm
                    {
                        \ifnum \cubecounty = \rowcountm
                        {
                            \foreach \ballcoord in \column
                            {
                                \drawballa{\ballcoord}
                            }
                        }
                        \fi
                    }
                    \fi
                }
            }
        }
        \fi
        }
    }
\end{tikzpicture} \end{tabular}}

\newcommand* \drawFamiliesEighteen{
\begin{tabular}{c} \begin{tikzpicture}[scale=0.8, every node/.style={scale=0.8}]
    \foreach \cubecountx in {0,1,2,3,4,5,6,7}
    {
        \foreach \cubecounty in {0,1}
        {
        \setcounter{end}{\cubecountx+10*\cubecounty};
        \ifnum \value{end} < 14
        {
        \drawcubeb{\cubecountx}{\cubecounty}
            \foreach \row [count=\rowcount, evaluate=\rowcount as \rowcountm using int(\rowcount-1)] in \familyEighteen
            {
                \foreach \column [count=\columncount, evaluate=\columncount as \columncountm using int(\columncount-1)] in \row
                {
                    \ifnum \cubecountx = \columncountm
                    {
                        \ifnum \cubecounty = \rowcountm
                        {
                            \foreach \ballcoord in \column
                            {
                                \drawballa{\ballcoord}
                            }
                        }
                        \fi
                    }
                    \fi
                }
            }
        }
        \fi
        }
    }
\end{tikzpicture} \end{tabular}}

\newcommand* \drawFamiliesNineteen{
\begin{tabular}{c} \begin{tikzpicture}[scale=0.8, every node/.style={scale=0.8}]
    \foreach \cubecountx in {0,1,2,3,4,5,6,7}
    {
        \foreach \cubecounty in {0}
        {
        \setcounter{end}{\cubecountx};
        \ifnum \value{end} < 4
        {
        \drawcubeb{\cubecountx}{\cubecounty}
            \foreach \row [count=\rowcount, evaluate=\rowcount as \rowcountm using int(\rowcount-1)] in \familyNineteen
            {
                \foreach \column [count=\columncount, evaluate=\columncount as \columncountm using int(\columncount-1)] in \row
                {
                    \ifnum \cubecountx = \columncountm
                    {
                        \ifnum \cubecounty = \rowcountm
                        {
                            \foreach \ballcoord in \column
                            {
                                \drawballa{\ballcoord}
                            }
                        }
                        \fi
                    }
                    \fi
                }
            }
        }
        \fi
        }
    }
\end{tikzpicture} \end{tabular}}

\newcommand* \drawFamiliesTwenty{
\begin{tabular}{c} \begin{tikzpicture}[scale=0.8, every node/.style={scale=0.8}]
    \foreach \cubecountx in {0,1,2,3,4,5,6,7}
    {
        \foreach \cubecounty in {0}
        {
        \drawcubeb{\cubecountx}{\cubecounty}
            \foreach \row [count=\rowcount, evaluate=\rowcount as \rowcountm using int(\rowcount-1)] in \familyTwenty
            {
                \foreach \column [count=\columncount, evaluate=\columncount as \columncountm using int(\columncount-1)] in \row
                {
                    \ifnum \cubecountx = \columncountm
                    {
                        \ifnum \cubecounty = \rowcountm
                        {
                            \foreach \ballcoord in \column
                            {
                                \drawballa{\ballcoord}
                            }
                        }
                        \fi
                    }
                    \fi
                }
            }
        }
    }
\end{tikzpicture} \end{tabular}}

\newcommand* \drawFamiliesTwentyOne{
\begin{tabular}{c} \begin{tikzpicture}[scale=0.8, every node/.style={scale=0.8}]
    \foreach \cubecountx in {0,1,2,3,4,5,6,7}
    {
        \foreach \cubecounty in {0}
        {
        \setcounter{end}{\cubecountx};
        \ifnum \value{end} < 1
        {
        \drawcubeb{\cubecountx}{\cubecounty}
            \foreach \row [count=\rowcount, evaluate=\rowcount as \rowcountm using int(\rowcount-1)] in \familyTwentyOne
            {
                \foreach \column [count=\columncount, evaluate=\columncount as \columncountm using int(\columncount-1)] in \row
                {
                    \ifnum \cubecountx = \columncountm
                    {
                        \ifnum \cubecounty = \rowcountm
                        {
                            \foreach \ballcoord in \column
                            {
                                \drawballa{\ballcoord}
                            }
                        }
                        \fi
                    }
                    \fi
                }
            }
        }
        \fi
        }
    }
\end{tikzpicture} \end{tabular}}

\begin{table*}
    \begin{tabular}{|@{ }c@{ }|@{ }c@{ }||@{ }c@{ }|@{ }c@{ }|}
        \hhline{--::--}
        Representative & Family & Representative & Family\\ \hhline{==::==}
        1. \drawFamilyOne & \drawFamiliesOne & 2. \drawFamilyThree & \drawFamiliesThree\\ \hhline{==::==}
        3. \drawFamilyFour & \drawFamiliesFour & 4. \drawFamilyTwo & \drawFamiliesTwo\\\hhline{==::==}
        5. \drawFamilySix \rule{0pt}{1.35cm} & \drawFamiliesSix & 6. \drawFamilyTwelve & \drawFamiliesTwelve \\\hhline{==::==}
        7. \drawFamilySeven &\drawFamiliesSeven & 8. \drawFamilyThirteen &\drawFamiliesThirteen\\\hhline{==::==}
        9. \drawFamilyFive \rule{0pt}{1.35cm} &\drawFamiliesFive & 10. \drawFamilyEight &\drawFamiliesEight\\\hhline{==::==}
        11. \drawFamilyNine &\drawFamiliesNine & 12. \drawFamilyFifteen &\drawFamiliesFifteen\\ \hhline{==::==}
        13. \drawFamilyEleven \rule{0pt}{1.35cm} &\drawFamiliesEleven &  14. \drawFamilyTen &\drawFamiliesTen\\ \hhline{==::==}
        15. \drawFamilyFourteen \rule{0pt}{1.35cm} &\drawFamiliesFourteen & 16. \drawFamilySeventeen &\drawFamiliesSeventeen \\ \hhline{==::==}
        17. \drawFamilySixteen \rule{0pt}{1.35cm} &\drawFamiliesSixteen & 18. \drawFamilyEighteen &\drawFamiliesEighteen \\ \hhline{==::==}
        19. \drawFamilyNineteen &\drawFamiliesNineteen & 20. \drawFamilyTwenty &\drawFamiliesTwenty \\ \hhline{==::==}
        21. \drawFamilyTwentyOne &\drawFamiliesTwentyOne &  &  \\ \hhline{--||--}
    \end{tabular}
    \caption{Classification of non-empty subsets of $\{0,1\}^3$ into families. All the sets in a family can be obtained from the representative by applying negations of bit values and/or permuting the coordinates.}
    \label{tab:cubestable}
\end{table*}

\end{document}